\documentclass[journal]{IEEEtran}

\usepackage{amsfonts}
\usepackage{amssymb}
\usepackage{amsthm}
\usepackage{amsmath,amsfonts,amssymb}
\usepackage[dvips]{graphicx}
\usepackage{subfigure}
\usepackage{verbatim}
\usepackage{setspace}
\usepackage{bm}
\usepackage{algorithmic} 
\usepackage[ruled,vlined]{algorithm2e}
\usepackage{cite}

\usepackage{changepage}
\usepackage{pdfpages}
\usepackage{color}
\newtheorem{theorem}{Theorem}

\allowdisplaybreaks

\newcommand{\figwidth}{7.8}
\IEEEoverridecommandlockouts

\begin{document}
\title{\huge Movable Antenna Enabled Near-Field Communications: Channel Modeling and Performance Optimization}
\author{Lipeng Zhu, ~\IEEEmembership{Member,~IEEE,}
		Wenyan Ma,~\IEEEmembership{Graduate Student Member,~IEEE,}
		Zhenyu Xiao,~\IEEEmembership{Senior Member,~IEEE,}
		and Rui Zhang,~\IEEEmembership{Fellow,~IEEE}
	\vspace{-0.3 cm}
	\thanks{L. Zhu and W. Ma are with the Department of Electrical and Computer Engineering, National University of Singapore, Singapore 117583 (e-mail: zhulp@nus.edu.sg, wenyan@u.nus.edu).}
	\thanks{Z. Xiao is with the School of Electronic and Information Engineering, Beihang University, Beijing, China 100191. (e-mail: xiaozy@buaa.edu.cn)}
	\thanks{R. Zhang is with School of Science and Engineering, Shenzhen Research Institute of Big Data, The Chinese University of Hong Kong, Shenzhen, Guangdong 518172, China (e-mail: rzhang@cuhk.edu.cn). He is also with the Department of Electrical and Computer Engineering, National University of Singapore, Singapore 117583 (e-mail: elezhang@nus.edu.sg).}
}

\maketitle


\begin{abstract}	
	Movable antenna (MA) technology offers promising potential to enhance wireless communication by allowing flexible antenna movement. To maximize spatial degrees of freedom (DoFs), larger movable regions are required, which may render the conventional far-field assumption for channels between transceivers invalid. In light of it, we investigate in this paper MA-enabled near-field communications, where a base station (BS) with multiple movable subarrays serves multiple users, each equipped with a fixed-position antenna (FPA). First, we extend the field response channel model for MA systems to the near-field propagation scenario. Next, we examine MA-aided multiuser communication systems under both digital and analog beamforming architectures. For digital beamforming, spatial division multiple access (SDMA) is utilized, where an upper bound on the minimum signal-to-interference-plus-noise ratio (SINR) across users is derived in closed form. A low-complexity algorithm based on zero-forcing (ZF) is then proposed to jointly optimize the antenna position vector (APV) and digital beamforming matrix (DBFM) to approach this bound. For analog beamforming, orthogonal frequency division multiple access (OFDMA) is employed, and an upper bound on the minimum signal-to-noise ratio (SNR) among users is derived. An alternating optimization (AO) algorithm is proposed to iteratively optimize the APV, analog beamforming vector (ABFV), and power allocation until convergence. For both architectures, we further explore MA design strategies based on statistical channel state information (CSI), with the APV updated less frequently to reduce the antenna movement overhead. Simulation results demonstrate that our proposed algorithms achieve performance close to the derived bounds and also outperform the benchmark schemes using dense or sparse arrays with FPAs.
\end{abstract}

%
\IEEEpeerreviewmaketitle

\section{Introduction}
\IEEEPARstart{T}{he} advance of multiple-input multiple output (MIMO) technologies has significantly improved the capacity and reliability of modern wireless communication systems \cite{Paulraj2004Anover}. As the number of mobile devices and the volume of data traffic increase exponentially, more antennas have been equipped at base stations (BSs) to fulfill the communication requirements of wireless networks. In this context, MIMO has evolved towards massive MIMO \cite{Larsson2014massive,you2024next} and extremely large-scale MIMO (XL-MIMO) \cite{lu2024nearXL,Wang2024XLMIMO} for attaining higher beamforming gains and enhanced spatial multiplexing. However, deploying hundreds or even thousands of antennas as well as their associated radio frequency (RF) front ends at the BS renders exorbitant hardware cost, energy consumption, and signal processing overhead, which hinders the efficiency of future wireless networks.

To address the limitations of existing MIMO systems, movable antenna (MA) technology has gained increasing attention recently as a promising solution to enhance wireless communication performance through flexible antenna movement \cite{zhu2023MAMag}, which is also known as fluid antenna system (FAS) \cite{wong2022bruce} or flexible-position MIMO \cite{zheng2024flexible} from the perspective of flexibility in antenna position. Unlike conventional fixed-position antennas (FPAs), MAs can adjust their positions dynamically within specified transmitter (Tx) or receiver (Rx) regions, allowing them to fully exploit the spatial variations of wireless channels in real-time. This adaptability offers several key advantages, including the potential to significantly improve spatial diversity, beamforming gains, and spatial multiplexing performance. As such, MA systems can achieve superior communication performance using the same or even fewer antennas compared to traditional FPA-based systems \cite{zhu2022MAmodel,ma2022MAmimo,Wong2023opport,New2024fluid}.

The implementation of MAs/FASs can be mainly categorized into two types, including the mechanically driven and electronically driven methods \cite{ning2024movable}. Mechanically driven MAs adopt drive components to physically change the position of the antenna. For example, a motor-based MA architecture was proposed in \cite{zhu2023MAMag} to realize antenna movement in a three-dimensional (3D) space, which was then extended to the six-dimensional (6D) movement by further incorporating the 3D antenna rotation, i.e., 6DMA \cite{shao20246DMANet,shao20246DMA}. Besides, the fluid/liquid antenna is another way of implementing mechanically driven MAs, which usually adopts liquid materials as radiating elements confined to a container for enabling one-dimensional (1D) or two-dimensional (2D) movement \cite{paracha2019liquid}. In general, mechanically driven MAs have a low hardware cost but the response speed may be limited by the mechanical movement delay. In comparison, electronically driven methods can realize equivalent antenna movement through electronic reconfiguration, such as displaced phase center antennas (DPCAs) \cite{Mitha2022DPCA} and dense array antennas \cite{Alrabadi2013dense}, which are beyond the original scope of MAs and can be referred to as specific implementations of reconfigurable antennas or FAS if we follow its definition in \cite{wong2022bruce}. Despite the faster response speed of electronically driven antennas, the hardware cost is also higher, which scales with the number of antennas and their degree of reconfigurability. Although there are different ways of implementing MAs/FASs in practice, they share similar mathematical models in terms of antenna movement and propagation channels \cite{zhu2024historical}. 

Various studies have demonstrated the performance advantages of MA-aided communication systems. For example, the spatial diversity gain of a single MA for increasing the received signal power in a narrowband communication system has been explored in \cite{zhu2022MAmodel,New2024fluid}, which was further extended to the MA-aided wideband systems \cite{zhu2024wideband} and MA-aided multiple-input single-output (MISO) systems \cite{mei2024movable}. Moreover, the spatial multiplexing performance of MA-aided MIMO systems was characterized in \cite{ma2022MAmimo,New2024MIMOFAS,chen2023joint}. The MA-assisted multiuser communication has also been widely investigated \cite{Wong2023opport,zhu2023MAmultiuser,xiao2023multiuser,wu2024globallyMA,hu2024power,Yang2024movable,zhou2024MANOMA}, where the multiuser interference can be efficiently mitigated by the joint optimization of MAs' positions and precoder/combiner. In addition, the MA array-enabled beamforming has been studied from the aspect of flexible interference nulling and multi-beam forming \cite{zhu2023MAarray,ma2024multi}. Moreover, the applications of MA arrays in satellite communication systems \cite{zhu2024dynamic}, unmanned aerial vehicle (UAV) communications \cite{liu2024MAUAV}, secure transmission \cite{hu2024secure,tang2024secure}, full-duplex communication \cite{ding2024fullduplex}, wireless sensing \cite{ma2024sensing}, and integrated sensing and communications (ISAC) \cite{lyu2024flexible} were also explored.

However, all of the aforementioned works \cite{zhu2022MAmodel,Wong2023opport,New2024fluid,zhu2024wideband,mei2024movable,ma2022MAmimo,New2024MIMOFAS,chen2023joint,zhu2023MAmultiuser,xiao2023multiuser,wu2024globallyMA,hu2024power,Yang2024movable,zhou2024MANOMA,zhu2023MAarray,ma2024multi,zhu2024dynamic,liu2024MAUAV,hu2024secure,tang2024secure,ding2024fullduplex,ma2024sensing,lyu2024flexible} on MA systems assume far-field channel conditions. In practice, to exploit more degrees of freedom (DoFs) in the spatial domain, larger regions are required for antenna movement. As the size of the movement region and/or the carrier frequency increase, the corresponding Rayleigh distance also increases, which potentially invalidates the far-field assumption in future wireless systems. For example, at 28 GHz, the Rayleigh distance for a 100-wavelength region is approximately 200 meters \cite{liu2023near}, necessitating the near-field modeling for MA systems. In this context, recent works have begun addressing near-field MA communications. Specifically, the authors in \cite{chen2024joint} investigated the MIMO communications enabled by MAs under the near-field condition, where the positions of receive MAs were jointly optimized with the transmit covariance matrix. Besides, the authors in \cite{zhu2024suppressing} investigated the MISO beamforming under near-field line-of-sight (LoS) channels, where the MAs' positions were optimized to alleviate the beam squint effect for wideband systems. However, both studies focus on point-to-point transmission and do not address multiple access scenarios. Another recent work \cite{ding2024near} explored the MA-aided near-field multiuser communications, where the MAs are employed at both the BS and users for saving the downlink transmit power subject to a minimum rate requirement for each user. In particular, a two-loop dynamic neighborhood pruning particle swarm optimization (DNPPSO) algorithm was developed to obtain suboptimal solutions for the positions of MAs. Nevertheless, this approach applies only to digital beamforming systems and does not reveal the performance upper bound for MA systems under the near-field channel condition.

Different from \cite{chen2024joint,zhu2024suppressing,ding2024near}, we investigate in this paper MA array-enabled near-field communications, where multiple movable subarrays at the BS serve a set of distributed users each equipped with a single FPA. In particular, we characterize the multiuser communication performance under both the digital beamforming and analog beamforming architectures. The main contributions of this paper are summarized as follows:
\begin{itemize}
	\item We extend the field response channel model for MA systems to near-field conditions by adopting the spherical wave model to characterize the spatial variation of LoS and non-LoS (NLoS) channel paths. In particular, the antennas in each subarray are collocated and moved together by the same drive, which can achieve a flexible trade-off between the communication performance and hardware cost. Then, we investigate MA-aided multiuser communications under both digital beamforming and analog beamforming architectures.
	\item Specifically, for digital beamforming, we adopt spatial division multiple access (SDMA) and derive an upper bound on the minimum signal-to-interference-plus-noise ratio (SINR) among users in closed form. A low-complexity algorithm based on zero-forcing (ZF) is then developed to jointly optimize the antenna position vector (APV) and digital beamforming matrix (DBFM) for approaching this upper bound. We also propose a statistical channel state information (CSI)-based APV design strategy to reduce antenna movement overhead.
	\item On the other hand, for analog beamforming, orthogonal frequency division multiple access (OFDMA) is adopted and an upper bound on the minimum receive signal-to-noise ratio (SNR) among users is derived in closed form. An alternating optimization (AO)-based algorithm is then developed to iteratively optimize the APV, analog beamforming vector (ABFV), and power allocation until convergence. Furthermore, the extension to statistical CSI-based antenna position optimization is also provided to reduce antenna movement overhead.
	\item Finally, simulations are conducted to evaluate the performance of the considered MA-enabled near-field multiuser communication systems. Numerical results show that our proposed algorithms can approach the derived performance upper bounds. It is also shown that MA systems significantly outperform other benchmark schemes based on dense or spare arrays with FPAs. Moreover, our proposed statistical CSI-based antenna position optimization can achieve a performance comparable to that based on instantaneous CSI, which provides a viable practical solution for deploying MAs at BSs.
\end{itemize}

The remainder of this paper is organized as follows. Section II presents the system and channel models. Section III deals with performance analysis and optimization for MA systems using digital beamforming, while Section IV considers performance characterization and optimization under analog beamforming. Section V provides the simulation results, and conclusions are finally drawn in Section VI.

\textit{Notation}: $a$, $\mathbf{a}$, $\mathbf{A}$, and $\mathcal{A}$ denote a scalar, a vector, a matrix, and a set, respectively. $(\cdot)^{*}$, $(\cdot)^{\rm{T}}$, $(\cdot)^{\rm{H}}$, and $(\cdot)^{\dagger}$ denote the conjugate, transpose, conjugate transpose, and Moore-Penrose inverse, respectively. $[\mathbf{a}]_{n}$ denotes the $n$-th entry of vector $\mathbf{a}$. $[\mathbf{A}]_{i,j}$ denotes the  entry in row $i$ and column $j$ of matrix $\mathbf{A}$. $\mathcal{A} \setminus \mathcal{B}$ and $\mathcal{A} \cap \mathcal{B}$ represent the subtraction set and intersection set of $\mathcal{A}$ and $\mathcal{B}$, respectively. $\mathbb{R}^{M \times N}$ and $\mathbb{C}^{M \times N}$ represent the sets of real and complex matrices/vectors of dimension $M \times N$, respectively. $\mathbb{Z}$ denotes the set of integer. $\Re(\cdot)$, $|\cdot|$, and $\angle(\cdot)$ denote the real part, the amplitude, and the phase of a complex number or complex vector, respectively. $\|\cdot\|_{1}$ and $\|\cdot\|_{2}$ denote the 1-norm and 2-norm of a vector, respectively. $\|\cdot\|_{\mathrm{F}}$ represents the Frobenius norm of a matrix. $\partial(\cdot)$ denotes the partial differential of a variable/function. $\mathbf{a} \sim \mathcal{CN}(\mathbf{0}, \mathbf{\Omega})$ indicates that $\mathbf{a}$ is a circularly symmetric complex Gaussian (CSCG) random vector with mean zero and covariance $\mathbf{\Omega}$. $\mathbb{E}\{\cdot\}$ is the expectation of a random variable. 

\section{System and Channel Model}
\begin{figure}[t]
	\begin{center}
		\includegraphics[width=8.8 cm]{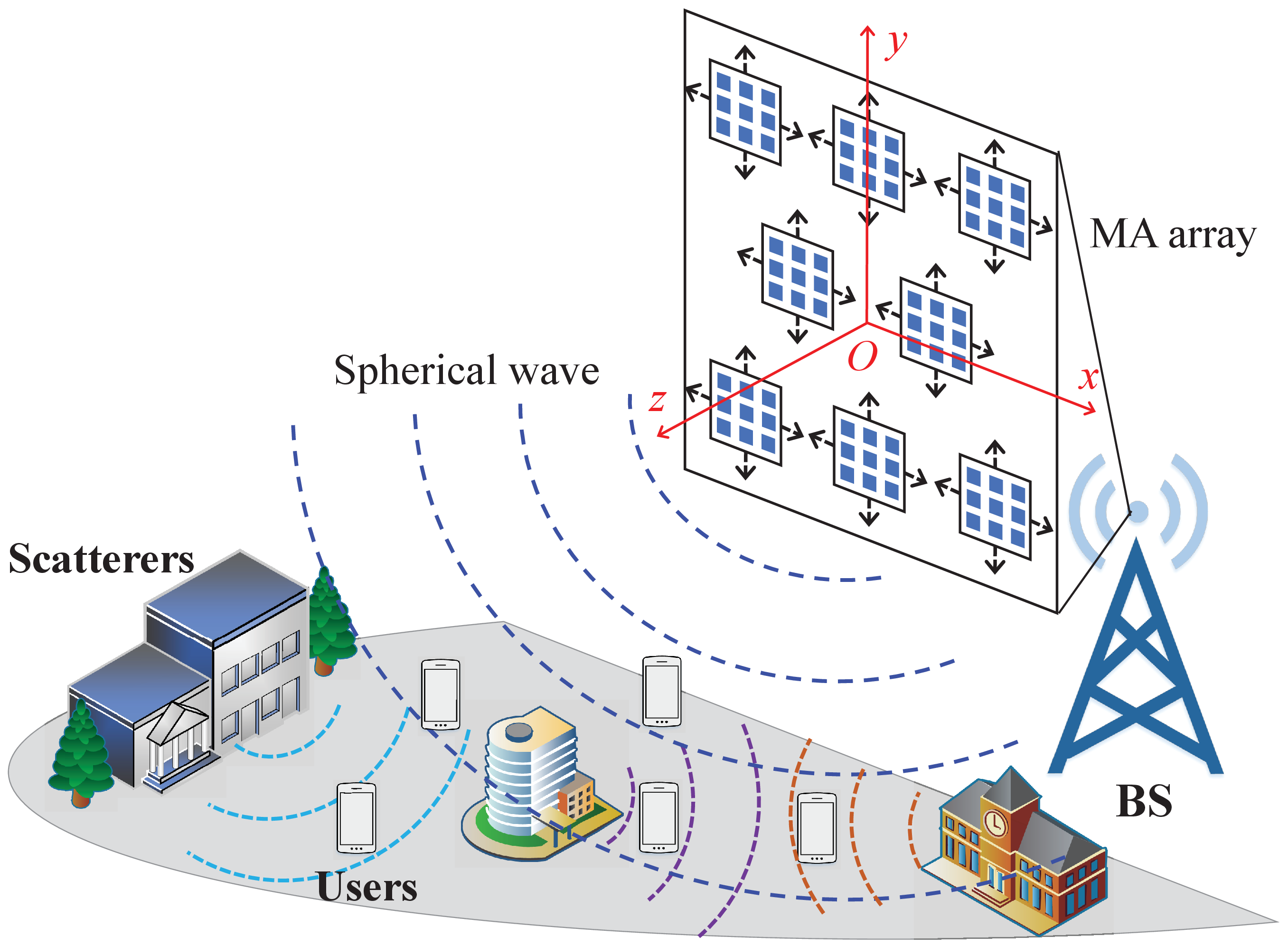}
		\caption{Illustration of the considered MA-enabled near-field communication system.}
		\label{fig:system}
	\end{center}
\end{figure}
As shown in Fig. \ref{fig:system}, the BS is equipped with a 2D MA array composed of $M$ movable subarrays. Without loss of generality, we assume that each MA subarray is a uniform planar array (UPA) consisting of $N=N_{x} \times N_{y}$ antennas, where $N_{x}$ and $N_{y}$ represent the number of antennas along horizontal and vertical directions, respectively. The number of users is denoted as $K$, with each user equipped with a single FPA. To depict the positions of MAs, we establish a 3D Cartesian coordinate system centered at the BS, where axes $x$ and $y$ are defined as the horizontal and vertical directions in the 2D MA array plane, respectively, while axis $z$ is perpendicular to the array plane. Denote $\mathbf{t}_{m} = [\tilde{\mathbf{t}}_{m}^{\mathrm{T}}, 0]^{\mathrm{T}}$, $1 \leq m \leq M$, and $\mathbf{q}_{n} = [\tilde{\mathbf{q}}_{n}^{\mathrm{T}}, 0]^{\mathrm{T}}$, $1 \leq n \leq N$, as the 3D position of the $m$-th MA subarray's center and the 3D position of the $n$-th antenna element relative to the subarray's center, respectively. 
Therein, $\tilde{\mathbf{t}}_{m}=[x_{m},y_{m}]^{\mathrm{T}}$ is the 2D position of the $m$-th MA subarray's center in the $x$-$O$-$y$ plane, which is confined in the antenna moving region, $\mathcal{C}=[-A/2,A/2]\times[-A/2,A/2]$. The constant vector $\tilde{\mathbf{q}}_{n}$ represents the 2D position of the $n$-th antenna relative to the subarray's center in the $x$-$O$-$y$ plane. Thus, the 3D position of the $n$-th antenna element in the $m$-th MA subarray is given by $\mathbf{t}_{m,n}=\mathbf{t}_{m} + \mathbf{q}_{n}$. 

The 3D coordinate of the $k$-th user is denoted as $\mathbf{s}_{k,0}$, $1 \leq k \leq K$. It is assumed that an LoS path exists between the BS and each user. The number of NLoS transmit paths from the BS to the $k$-th user is $L_{k}$, which corresponds to $L_{k}$ point scatterers around the BS. Specifically, the 3D coordinate of the $\ell$-th scatterer for user $k$'s channel is denoted as $\mathbf{s}_{k,\ell}$, $1 \leq \ell \leq L_{k}$. According to the spherical wave model, the near-field response vector (NFRV) for the channel from antenna position $\mathbf{t}_{m,n}$ at the BS to user $k$ is given by\footnote{For the case without LoS path, the NFRV reduces dimension to $L_{k}$, which characterizes the phase variation of all NLoS paths from the BS to each user.}
\begin{equation}\label{eq_FRV}
	\mathbf{g}_{k}(\mathbf{t}_{m,n})=\left[\mathrm{e}^{\mathrm{j}\frac{2\pi}{\lambda}\|\mathbf{t}_{m,n}-\mathbf{s}_{k,0}\|_{2}},\dots,
	\mathrm{e}^{\mathrm{j}\frac{2\pi}{\lambda}\|\mathbf{t}_{m,n}-\mathbf{s}_{k,L_{k}}\|_{2}}\right]^{\mathrm{T}},
\end{equation}
which characterizes the phase variation of all the channel paths between the BS and user $k$. Let $\mathbf{b}_{k} \in \mathbb{C}^{(L_{k}+1) \times 1}$ denote the path response vector (PRV) consisting of the complex coefficients of all the $(L_{k}+1)$ channel paths from the origin at the BS to user $k$. The channel gain from antenna position $\mathbf{t}_{m,n}$ at the BS to user $k$ is thus given by
\begin{equation}\label{eq_channel}
	h_{k}(\mathbf{t}_{m,n})=\mathbf{g}_{k}^{\mathrm{H}}(\mathbf{t}_{m,n})\mathbf{b}_{k}, 1 \leq k \leq K.
\end{equation}
Furthermore, we define $\tilde{\mathbf{t}} \triangleq [\tilde{\mathbf{t}}_{1}^{\mathrm{T}},\tilde{\mathbf{t}}_{2}^{\mathrm{T}},\dots, \tilde{\mathbf{t}}_{M}^{\mathrm{T}}]^{\mathrm{T}} \in \mathbb{R}^{2M \times 1}$ as the APV of all the $M$ MA subarrays. The near-field response matrix (NFRM) for user $k$ is then defined as $\mathbf{G}_{k}(\tilde{\mathbf{t}}) \triangleq [\mathbf{g}_{k}(\mathbf{t}_{1,1}), \dots, \mathbf{g}_{k}(\mathbf{t}_{1,N}), \dots, \mathbf{g}_{k}(\mathbf{t}_{M,1}), \dots, \mathbf{g}_{k}(\mathbf{t}_{M,N})] \in \mathbb{C}^{(L_{k}+1) \times MN}$. As such, the channel vector between the MA array of the BS and the FPA of user $k$ can be expressed as a function of the APV, i.e., 
\begin{equation}\label{eq_channel_vec}
	\mathbf{h}_{k}(\tilde{\mathbf{t}})=\mathbf{G}_{k}^{\mathrm{H}}(\tilde{\mathbf{t}})\mathbf{b}_{k}, 1 \leq k \leq K.
\end{equation}

Next, we investigate the MA-enabled near-field communication systems under the digital beamforming and analog beamforming architectures, respectively.

\section{MA Systems with Digital Beamforming}
\subsection{Problem Formulation}
For the considered MA-enabled near-field communication system under the digital beamforming architecture, multiple users are served by the BS via SDMA in the downlink to fully exploit the spatial multiplexing gain. Specifically, denoting $\mathbf{H}(\tilde{\mathbf{t}})=[\mathbf{h}_{1}(\tilde{\mathbf{t}}),\mathbf{h}_{2}(\tilde{\mathbf{t}}),\dots,\mathbf{h}_{K}(\tilde{\mathbf{t}})] \in \mathbb{C}^{MN \times K}$ as the channel matrix between the BS and all users, the received signals at the users can be expressed as
\begin{equation}\label{eq_signal_dig}
	\mathbf{y} = \mathbf{H}^{\mathrm{H}}(\tilde{\mathbf{t}})\mathbf{W}\mathbf{s} + \mathbf{n},
\end{equation}
where $\mathbf{W} = [\mathbf{w}_{1}, \mathbf{w}_{2}, \dots, \mathbf{w}_{K}] \in \mathbb{C}^{MN \times K}$ is the DBFM at the BS. $\mathbf{s} \in \mathbb{C}^{K \times 1}$ denotes the normalized CSCG signals satisfying $\mathbb{E}\{\mathbf{s}\mathbf{s}^{\mathrm{H}}\} = \mathbf{I}_{K}$. $\mathbf{n} \sim \mathcal{CN}(\mathbf{0}, \sigma^{2}\mathbf{I}_{K})$ represents the additive Gaussian noise at the users with average power $\sigma^{2}$.

Then, the receive SINR for user $k$ is given by
\begin{equation}\label{eq_SINR}
	\gamma_{k}(\tilde{\mathbf{t}},\mathbf{W}) = \frac{\left|\mathbf{h}_{k}^{\mathrm{H}}(\tilde{\mathbf{t}}) \mathbf{w}_{k}\right|^{2}}{\sum \limits _{j=1, j \neq k}^{K} \left|\mathbf{h}_{k}^{\mathrm{H}}(\tilde{\mathbf{t}}) \mathbf{w}_{j}\right|^{2} + \sigma^{2}}.
\end{equation}
To guarantee user fairness, we aim to maximize the minimum SINR among multiple users by jointly optimizing the APV and DBFM for the MA array, which can be formulated as the following optimization problem:
\begin{subequations}\label{eq_problem_DBF}
	\begin{align}
		&\mathop{\max}\limits_{\tilde{\mathbf{t}}, \mathbf{W}}~
		\min \limits_{1 \leq k \leq K} ~ \gamma_{k}(\tilde{\mathbf{t}},\mathbf{W}) \label{eq_problem_DBF_a}\\
		&~~\mathrm{s.t.}~~  \left\|\mathbf{W}\right\|_{\mathrm{F}}^{2} \leq P, \label{eq_problem_DBF_b}\\
		&~~~~~~~~~  \tilde{\mathbf{t}}_{m} \in \mathcal{C},~ 1 \leq m \leq M, \label{eq_problem_DBF_c}\\
		&~~~~~~~~  \left\|\tilde{\mathbf{t}}_{m}-\tilde{\mathbf{t}}_{\hat{m}}\right\|_{2} \geq d_{\min},~ 1 \leq m \neq \hat{m} \leq M, \label{eq_problem_DBF_d}
	\end{align}
\end{subequations}
where constraint \eqref{eq_problem_DBF_b} indicates that the total transmit power of the BS should not exceed the maximum value, $P$; constraint \eqref{eq_problem_DBF_c} confines the moving region of each MA subarray; and constraint \eqref{eq_problem_DBF_d} guarantees that the distance between any two MA subarrays should be no less than the minimum threshold, $d_{\min}$, which can avoid their coupling and/or overlap. Since problem \eqref{eq_problem_DBF} is a non-convex optimization problem with coupled variables, it is difficult to obtain globally optimal solutions for it in a polynomial time. Thus, we first analyze the upper bound on the objective function of problem \eqref{eq_problem_DBF} and illustrate the condition to achieve the performance bound under the special case of a single (LoS/NLoS) path for each user. Then, we develop a low-complexity algorithm to obtain suboptimal solutions for problem \eqref{eq_problem_DBF}.

\subsection{Performance Analysis}
If the received signal power for each user is equal to its maximum value and the interference between any two users is equal to zero, the SINR for each user in \eqref{eq_SINR} achieves an upper bound. Specifically, the maximum received signal power for each user can be achieved by maximal ratio transmitting (MRT), i.e., $\mathbf{w}_{k}=\sqrt{p}_{k}\mathbf{h}_{k}(\tilde{\mathbf{t}})/\|\mathbf{h}_{k}(\tilde{\mathbf{t}})\|_{2}$, where $p_{k}$ represents the transmit power allocated to the signal for user $k$. Then, the SINR in \eqref{eq_SINR} is upper-bounded by the SNR as
\begin{equation}\label{eq_SINR_bound}
	\gamma_{k}(\tilde{\mathbf{t}},\mathbf{W}) \leq \frac{\left|\mathbf{h}_{k}^{\mathrm{H}}(\tilde{\mathbf{t}}) \mathbf{w}_{k}\right|^{2}}{\sigma^{2}} = \frac{\left\|\mathbf{h}_{k}(\tilde{\mathbf{t}})\right\|_{2}^{2} p_{k}}{\sigma^{2}},~1 \leq k \leq K,
\end{equation}
where the equality holds if the channel vectors for all users are orthogonal to each other, i.e., $\mathbf{h}_{k}^{\mathrm{H}}(\tilde{\mathbf{t}})\mathbf{h}_{\hat{k}}(\tilde{\mathbf{t}}) = 0$, $k \neq \hat{k}$. Moreover, the channel power gain for user $k$ is upper-bounded by
\begin{equation}\label{eq_channel_power}
	\left\|\mathbf{h}_{k}(\tilde{\mathbf{t}})\right\|_{2}^{2} = \sum \limits_{m=1}^{M} \sum \limits_{n=1}^{N} \left|\mathbf{g}_{k}^{\mathrm{H}}(\mathbf{t}_{m,n})\mathbf{b}_{k}\right|^{2} \leq MN\left\|\mathbf{b}_{k}\right\|_{1}^{2},
\end{equation}
where the equality holds if the phase of the NFRV aligns with that of the PRV, i.e., $\mathbf{g}_{k}(\mathbf{t}_{m,n}) = \mathrm{e}^{\mathrm{j} \angle \mathbf{b}_{k}}$.

Substituting \eqref{eq_channel_power} into \eqref{eq_SINR_bound}, we obtain an upper bound on the SINR for user $k$ as
\begin{equation}\label{eq_SINR_bound2}
	\gamma_{k}(\tilde{\mathbf{t}},\mathbf{W}) \leq \frac{MN\left\|\mathbf{b}_{k}\right\|_{1}^{2} p_{k}}{\sigma^{2}},~1 \leq k \leq K.
\end{equation}
To maximize the minimum SINR among all the $K$ users, the transmit power should be appropriately allocated such that all the users have the same SINR. Otherwise, we can always reduce the power for the user with higher SINR and reallocate it to the users with lower SINR. Thus, the optimal power allocation for maximizing the minimum SINR can be obtained by solving the following equations:
\begin{equation}\label{eq_power_allo_equ}
	\left\{\begin{aligned}
		&\frac{MN\left\|\mathbf{b}_{k}\right\|_{1}^{2} p_{k}}{\sigma^{2}} = \bar{\gamma}, ~1 \leq k \leq K,\\
		&\sum \limits_{k=1}^{K} p_{k} = P,
	\end{aligned}
	\right.
\end{equation}
where $\bar{\gamma}$ denotes the upper bound on the minimum SINR in \eqref{eq_problem_DBF_a} and it can be derived in closed form as
\begin{equation}\label{eq_minSINR_bound}
	\bar{\gamma} = \frac{P}{\sigma^{2}}\left(\sum \limits_{k=1}^{K} \frac{1}{MN\left\|\mathbf{b}_{k}\right\|_{1}^{2}}\right)^{-1}.
\end{equation}

It is worth noting that the upper bound in \eqref{eq_minSINR_bound} is not always achievable because the optimization of the APV, $\tilde{\mathbf{t}}$, cannot always guarantee the orthogonality between channel vectors and the phase alignment between the NFRV and PRV for all users. Nevertheless, if all the users only have a single (LoS/NLoS) channel path and the size of each subarray is $N=1$, we can derive a sufficient and necessary condition for the APV to achieve the SINR upper bound. For notation simplicity, denote $\mathbf{s}_{k}$ and $b_{k}$ as the 3D position of the user/scatterer corresponding to the single LoS/NLoS path and its coefficient, respectively. Then, we have the following theorem. 
\begin{theorem}\label{Theo_dig}
	Under the condition of a single channel path for each user and $N=1$, the upper bound on the minimum SINR in \eqref{eq_minSINR_bound} is achieved if and only if the 3D position of each MA subarray belongs to the following set: 
	\begin{equation}\label{eq_opt_APV_dig}
		\begin{aligned}
			\mathbf{t}_{m} \in &\Big{\{} \mathbf{t} \big{\arrowvert} \left\|\mathbf{t}-\mathbf{s}_{k}\right\|_{2} -\left\|\mathbf{t}-\mathbf{s}_{\hat{k}}\right\|_{2} = \lambda(n_{k,\hat{k},m}+\phi_{k,\hat{k},m}),\\
			&\forall k \neq \hat{k}, \forall n_{k,\hat{k},m} \in \mathbb{Z}, 0 \leq \phi_{k,\hat{k},m} < 1 \Big{\}} \triangleq \mathcal{T}_{m},
		\end{aligned}
	\end{equation}
	with $\sum_{m=1}^{M} e^{\mathrm{j}2\pi\phi_{k,\hat{k},m}}=0$.
\end{theorem}
\begin{proof}
	See Appendix \ref{Appendix_dig}.
\end{proof}
Theorem \ref{Theo_dig} specifies the structure of the APV solution for achieving the SINR upper bound under the condition of a single channel path for each user, where $\mathcal{T}_{m}$ is composed of the intersections of a series of hyperbolas determined by parameters $n_{k,\hat{k},m}$ and $\phi_{k,\hat{k},m}$. It can be used to obtain an optimal APV by checking the intersection of $\{\mathcal{T}_{m}\}_{m=1}^{M}$ and constraints \eqref{eq_problem_DBF_c} and \eqref{eq_problem_DBF_d}. For example, for the case of $K=2$, set $\mathcal{T}_{m}$ is simplified as
	\begin{equation}
	\begin{aligned}
		\mathcal{T}_{m}= &\Big{\{} \mathbf{t} \big{\arrowvert} \left\|\mathbf{t}-\mathbf{s}_{1}\right\|_{2} -\left\|\mathbf{t}-\mathbf{s}_{2}\right\|_{2} = \lambda(n_{1,2,m}+\phi_{1,2,m}),\\
		&~~~~~~~~~~~~~~~~~~~~~~~\forall n_{1,2,m} \in \mathbb{Z}, 0 \leq \phi_{1,2,m} < 1 \Big{\}},
	\end{aligned}
\end{equation}
which is a series of hyperbolas w.r.t. focuses $\mathbf{s}_{1}$ and $\mathbf{s}_{2}$. Consider a simple setting of $\phi_{1,2,m} = m/M$, $m=1,\dots,M$, satisfying $\sum_{m=1}^{M} e^{\mathrm{j}2\pi\phi_{1,2,m}}=0$. Then, we only need to select a point $\mathbf{t}_{m}$ in $\mathcal{T}_{m} \cap \mathcal{C}$ satisfying the inter-subarray distance constraint \eqref{eq_problem_DBF_d}, which yields an optimal solution for the APV. However, as the number of users increases, the computational complexity for finding the optimal APV exponentially increases with $K$ because we need to search $n_{k,\hat{k},m}$ and $\phi_{k,\hat{k},m}$ for all users. Besides, for certain channel conditions, $\mathcal{T}_{m}$ may be an empty set and thus the upper bound in \eqref{eq_minSINR_bound} is not achievable. Moreover, for the general case of multiple channel paths for each user or multiple antennas in each subarray, it is difficult to derive an explicit solution for the APV. Due to these reasons, we develop a general algorithm with a low complexity to numerically find suboptimal solutions for the APV and DBFM in the next subsection. 

\subsection{Optimization Algorithm}
Due to the sufficiently large size of the antenna moving region, the correlation of channel vectors for multiple users can be significantly reduced by antenna position optimization. Thus, we adopt the ZF method for designing the DBFM. Specifically, the ZF-based DBFM for maximizing the minimum SINR among multiple users is given by
\begin{equation}\label{eq_ZF}
	\mathbf{W}_{\mathrm{ZF}}(\tilde{\mathbf{t}})= \sqrt{P}\frac{\left(\mathbf{H}(\tilde{\mathbf{t}})^{\mathrm{H}}\right)^{\dagger}}{\left\|\left(\mathbf{H}(\tilde{\mathbf{t}})^{\mathrm{H}}\right)^{\dagger}\right\|_{\mathrm{F}}} = \frac{\mathbf{H}(\tilde{\mathbf{t}})\left(\mathbf{H}^{\mathrm{H}}(\tilde{\mathbf{t}})\mathbf{H}(\tilde{\mathbf{t}})\right)^{-1}}{\left\| \mathbf{H}(\tilde{\mathbf{t}})\left(\mathbf{H}^{\mathrm{H}}(\tilde{\mathbf{t}})\mathbf{H}(\tilde{\mathbf{t}})\right)^{-1} \right\|_{\mathrm{F}}},
\end{equation}
which yields the same SINR for all $K$ users as 
\begin{equation}\label{eq_SINR_min}
	\gamma(\tilde{\mathbf{t}}) = \frac{P}{\left\| \mathbf{H}(\tilde{\mathbf{t}})\left(\mathbf{H}^{\mathrm{H}}(\tilde{\mathbf{t}})\mathbf{H}(\tilde{\mathbf{t}})\right)^{-1} \right\|_{\mathrm{F}}^{2} \sigma^{2}} = \frac{P}{\mathrm{tr}\{\mathbf{Z}(\tilde{\mathbf{t}})^{-1}\} \sigma^{2}},
\end{equation}
with $\mathbf{Z}(\tilde{\mathbf{t}}) \triangleq \mathbf{H}^{\mathrm{H}}(\tilde{\mathbf{t}})\mathbf{H}(\tilde{\mathbf{t}})$. Then, the maximization of the minimum SINR is simplified as the following problem:
\begin{equation}\label{eq_problem_DBF2}
	\begin{aligned}
		&\mathop{\max}\limits_{\tilde{\mathbf{t}}}~\gamma(\tilde{\mathbf{t}}) \\
		&~~\mathrm{s.t.}~~  \eqref{eq_problem_DBF_c},~\eqref{eq_problem_DBF_d}.
	\end{aligned}
\end{equation}
Note that problem \eqref{eq_problem_DBF2} is not equivalent to problem \eqref{eq_problem_DBF} because the ZF beamforming is not optimal. Nevertheless, we will show in the simulation that the ZF-based solution can approach the performance upper bound given by \eqref{eq_SINR_bound} because the antenna position optimization in a sufficiently large region can significantly reduce the correlation of channel vectors for multiple users.

To solve the non-convex problem \eqref{eq_problem_DBF2} efficiently, we adopt the projected gradient ascent method. Specifically, denote the APV obtained in the $(i-1)$-th iteration as $\tilde{\mathbf{t}}^{(i-1)}$ and the gradient of $\gamma(\tilde{\mathbf{t}})$ w.r.t. $\tilde{\mathbf{t}}$ as $\nabla_{\tilde{\mathbf{t}}} \gamma(\tilde{\mathbf{t}})$, which can be derived in closed form shown in Appendix B. For the $i$-th iteration, the APV is updated as
\begin{equation}\label{eq_APV_update}
	\tilde{\mathbf{t}}^{(i)} = \mathcal{B} \left\{ \tilde{\mathbf{t}}^{(i-1)} + \tau^{(i)} \nabla_{\tilde{\mathbf{t}}} \gamma(\tilde{\mathbf{t}}^{(i-1)}) \right\},
\end{equation}
where $\mathcal{B}\{\cdot\}$ is a projection function. If an element in $\tilde{\mathbf{t}}$ exceeds the feasible region, it is then projected to the nearest boundary of its feasible region \cite{zhu2023MAmultiuser}. $\tau^{(i)}$ is the step size, which is obtained by backtracking line search. Specifically, for each iteration, the step size is initialized as a large positive value, $\tau^{(i)} = \tau$. Then, we repeatedly shrink the step size by a factor $\mu \in (0,1)$, i.e., $\tau^{(i)} \leftarrow \mu \tau^{(i)}$, until constraint \eqref{eq_problem_DBF_d} and the Armijo–Goldstein condition are both satisfied \cite{boyd2004convex}, i.e.,
\begin{subequations}\label{eq_Armijo–Goldstein}
	\begin{align}
		&\gamma(\tilde{\mathbf{t}}^{(i)}) \geq \gamma(\tilde{\mathbf{t}}^{(i-1)}) + \xi \tau^{(i-1)} \left\|\nabla_{\tilde{\mathbf{t}}} \gamma(\tilde{\mathbf{t}}^{(i-1)})\right\|_{2}^{2},\\
		&\left\|\tilde{\mathbf{t}}_{m}^{(i)}-\tilde{\mathbf{t}}_{\hat{m}}^{(i)}\right\|_{2} \geq d_{\min},~ 1 \leq m \neq \hat{m} \leq M,
	\end{align}
\end{subequations}
where $\xi \in (0,1)$ is a predefined parameter to control the incremental speed of the objective function. 

The overall algorithm for solving problem \eqref{eq_problem_DBF} is summarized in Algorithm \ref{alg_DBF}. In line 1, the APV is initialized by letting the subarrays uniformly distributed within the entire antenna moving region. For each iteration, the backtracking line search for the step size is operated in lines 7-10. The algorithm terminates if the increment of the objective function is below a predefined threshold, $\epsilon_{1}$, or it reaches the maximum number of iterations, $I_{1}$. Since the objective function is non-decreasing over the iterations and it is upper-bounded, the convergence of Algorithm \ref{alg_DBF} is guaranteed. The computational complexity of calculating the gradient in line 4 is $\mathcal{O}(MNL^{2}K^{2})$, with $L=\max \limits_{1 \leq k \leq K}\{L_{k}\}$. The complexity of calculating the objective function is $\mathcal{O}(MNLK+MNK^2)$, which entails the complexity of $\mathcal{O}(J_{1}MNK(L+K))$ in lines 7-10, with $J_{1}$ being the maximum number of iterations for backtracking line search. Besides, the complexity of calculating the ZF-based DBFM in line 16 is $\mathcal{O}(MNK^2)$. Thus, the total computational complexity of Algorithm \ref{alg_DBF} is $\mathcal{O}(I_{1}MNK(J_{1}L+J_{1}K+L^{2}K))$.

\begin{algorithm}[t]\small
	\caption{ZF-based algorithm for problem \eqref{eq_problem_DBF}.}
	\label{alg_DBF}
	\begin{algorithmic}[1]
		\REQUIRE ~$M$, $N_{x}$, $N_{y}$, $K$, $P$, $\sigma^{2}$, $\lambda$, $A$, $d_{\min}$, $\{\mathbf{q}_{n}\}$, $\{L_{k}\}$,\\ ~~~~~~$\{\mathbf{s}_{k,\ell}\}$, $\{\mathbf{b}_{k}\}$, $\tau$, $\xi$, $\mu$, $\epsilon_{1}$, $I_{1}$.
		\ENSURE ~$\tilde{\mathbf{t}}^{\star}$ and $\mathbf{W}^{\star}$. \\
		\STATE Initialize the APV $\tilde{\mathbf{t}}^{(0)}$.
		\STATE Obtain $\gamma(\tilde{\mathbf{t}}^{(0)})$ according to \eqref{eq_SINR_min}.
		\FOR   {$i=1:I_{1}$}
		\STATE Calculate the gradient according to \eqref{eq_SINR_grad_x} and \eqref{eq_SINR_grad_y}.
		\STATE Initialize step size $\tau^{(i)}=\tau$.
		\STATE Update the APV $\tilde{\mathbf{t}}^{(i)}$ according to \eqref{eq_APV_update}.
		\WHILE {\eqref{eq_Armijo–Goldstein} is not satisfied}
		\STATE Shrink the step size $\tau^{(i)} \leftarrow \mu \tau^{(i)}$.
		\STATE Update the APV $\tilde{\mathbf{t}}^{(i)}$ according to \eqref{eq_APV_update}.
		\ENDWHILE
		\IF    {$\gamma(\tilde{\mathbf{t}}^{(i)})-\gamma(\tilde{\mathbf{t}}^{(i-1)})<\epsilon_{1}$}
		\STATE Break.
		\ENDIF
		\ENDFOR
		\STATE Set the APV as $\tilde{\mathbf{t}}^{\star}=\tilde{\mathbf{t}}^{(i)}$.
		\STATE Calculate the DBFM as $\mathbf{W}^{\star}=\mathbf{W}_{\mathrm{ZF}}(\tilde{\mathbf{t}}^{\star})$ according to \eqref{eq_ZF}.
		\RETURN $\tilde{\mathbf{t}}^{\star}$ and $\mathbf{W}^{\star}$.
	\end{algorithmic}
\end{algorithm}

\subsection{Statistical CSI Based Antenna Position Optimization}
Note that the joint optimization of the APV and DBFM in problem \eqref{eq_problem_DBF} requires the antenna movement based on instantaneous CSI. In practice, due to the mobility of users and scatterers, the wireless channels between the BS and users usually exhibit variations over time. However, the speed of antenna movement is limited, which may hinder the efficiency of MA arrays under time-varying or fast fading channels. To address this issue, we propose a two-timescale optimization scheme for MA arrays, where the APV is designed based on the statistical CSI over a long time period, while the DBFM is optimized based on instantaneous CSI. To this end, we formulate a two-timescale optimization problem for maximizing the ergodic SINR of $K$ users as follows:
\begin{subequations}\label{eq_problem_DBF_ergodic}
	\begin{align}
		&\mathop{\max}\limits_{\tilde{\mathbf{t}}}~ \mathbb{E}\left\{\mathop{\max}\limits_{\mathbf{W}}~
		\min \limits_{1 \leq k \leq K} ~ \gamma_{k}(\tilde{\mathbf{t}},\mathbf{W})\right\} \label{eq_problem_DBF_ergodic_a}\\
		&~~\mathrm{s.t.}~~  \eqref{eq_problem_DBF_b},~\eqref{eq_problem_DBF_c},~\eqref{eq_problem_DBF_d},
	\end{align}
\end{subequations}
where the expectation is operated over the random instantaneous channels\footnote{Problem \eqref{eq_problem_DBF_ergodic} represents a general formation that is applicable to the design of MA array geometry under any given statistical CSI. The modeling and acquisition of statistical CSI between the BS and users are beyond the scope of this paper and will be left for future work.}. Since it is difficult to derive the expected SINR in closed form, to facilitate the optimization, we approximate the expectation in \eqref{eq_problem_DBF_ergodic_a} by Monte Carlo simulations as
\begin{equation}\label{eq_SINR_ergodic}
	\mathbb{E}\left\{\mathop{\max}\limits_{\mathbf{W}}
	\min \limits_{1 \leq k \leq K} \gamma_{k}(\tilde{\mathbf{t}},\mathbf{W})\right\} \approx \frac{1}{Q_{1}} \sum \limits_{q=1}^{Q_{1}} \mathop{\max}\limits_{\mathbf{W}}
	\min \limits_{1 \leq k \leq K} \gamma_{k}^{q}(\tilde{\mathbf{t}},\mathbf{W}),
\end{equation}
where $Q_{1}$ is the total number of Monte Carlo simulations and $\gamma_{k}^{q}$ denotes the SINR for user $k$ under the $q$-th, $1 \leq q \leq Q_{1}$, independent channel realization based on the known statistical CSI (e.g., user/channel distribution). Note that for each channel realization, we can also adopt the ZF-based DBFM given by \eqref{eq_ZF}. Then, the APV can be optimized by projected gradient ascent similar to that in Algorithm \ref{alg_DBF}, whereas the gradient should be averaged over all the $Q_{1}$ independent channel realizations.

\section{MA Systems With Analog Beamforming}
\subsection{Problem Formulation}
For MA-aided analog beamforming systems, SDMA cannot be utilized in multiuser communications because a single RF chain can only support one spatial data stream. Nevertheless, multiple users can be served by the BS via OFDMA, where the MA array is able to increase the beam coverage performance. Denoting $\mathbf{w} \in \mathbb{C}^{MN \times 1}$ as the ABFV at the BS, the received signal at user $k$ can be expressed as
\begin{equation}\label{eq_signal_ana}
	y_{k} = \mathbf{h}_{k}^{\mathrm{H}}(\tilde{\mathbf{t}})\mathbf{w}\sqrt{p_{k}}s_{k} + n_{k},~1 \leq k \leq K,
\end{equation}
where $s_{k}$ denotes the normalized CSCG signal for user $k$ and $p_{k}$ is the transmit power. $n_{k}$ represents the additive Gaussian noise at the user with average power $\sigma^{2}$. The receive SNR for user $k$ is thus given by
\begin{equation}\label{eq_SNR}
	\eta_{k}(\tilde{\mathbf{t}},\mathbf{w}, \mathbf{p}) = \frac{\left|\mathbf{h}_{k}^{\mathrm{H}}(\tilde{\mathbf{t}}) \mathbf{w}\right|^{2} p_{k}} {\sigma^{2}},~1 \leq k \leq K,
\end{equation}
where $\mathbf{p} \triangleq [p_{1}, p_{2}, \dots, p_{K}]^{\mathrm{T}}$ is the power allocation vector.

To guarantee user fairness, we maximize the minimum SNR among multiple users by jointly optimizing the APV, ABFV, and power allocation at the BS, which can be formulated as the following optimization problem:
\begin{subequations}\label{eq_problem_ABF}
	\begin{align}
		&\mathop{\max}\limits_{\tilde{\mathbf{t}}, \mathbf{w}, \mathbf{p}}~
		\min \limits_{1 \leq k \leq K} ~ \eta_{k}(\tilde{\mathbf{t}},\mathbf{w}, \mathbf{p}) \label{eq_problem_ABF_a}\\
		&~~\mathrm{s.t.}~~  \left|\left[\mathbf{w}\right]_{n}\right| = \frac{1}{\sqrt{MN}}, ~ 1 \leq n \leq MN, \label{eq_problem_ABF_b}\\
		&~~~~~~~~   \left\|\mathbf{p}\right\|_{1} \leq P, \label{eq_problem_ABF_c}\\
		&~~~~~~~~~  \tilde{\mathbf{t}}_{m} \in \mathcal{C},~ 1 \leq m \leq M, \label{eq_problem_ABF_d}\\
		&~~~~~~~~  \left\|\tilde{\mathbf{t}}_{m}-\tilde{\mathbf{t}}_{\hat{m}}\right\|_{2} \geq d_{\min},~ 1 \leq m \neq \hat{m} \leq M, \label{eq_problem_ABF_e}
	\end{align}
\end{subequations}
where constraint \eqref{eq_problem_ABF_b} is the constant-modulus constraint on the ABFV; constraint \eqref{eq_problem_ABF_c} indicates that the total transmit power of the BS should be no large than $P$; and constraints \eqref{eq_problem_ABF_d} and \eqref{eq_problem_ABF_e} confine the positions of MA subarrays. Problem \eqref{eq_problem_ABF} is also non-convex with coupled variables. Similar to Section III, we first analyze the upper bound on the objective function of problem \eqref{eq_problem_ABF} and then develop a low-complexity algorithm to obtain suboptimal solutions.

\subsection{Performance Analysis}
For any given power allocation, the receive SNR for each user is maximized if the phase of the ABFV is aligned with that of the channel vector, i.e., $\mathbf{w} = \mathrm{e}^{\mathrm{j}\angle \mathbf{h}_{k}(\tilde{\mathbf{t}})}/\sqrt{MN}$, which yields an upper bound on the SNR for user $k$ as 
\begin{equation}\label{eq_SNR_bound}
	\eta_{k}(\tilde{\mathbf{t}},\mathbf{w}, \mathbf{p}) \leq \frac{\left\|\mathbf{h}_{k}(\tilde{\mathbf{t}})\right\|_{1}^{2} p_{k}}{MN\sigma^{2}},~1 \leq k \leq K.
\end{equation}
Moreover, the squared 1-norm of the channel vector is upper-bounded by
\begin{equation}\label{eq_channel_power_one}
	\left\|\mathbf{h}_{k}(\tilde{\mathbf{t}})\right\|_{1}^{2} = \left(\sum \limits_{m=1}^{M} \sum \limits_{n=1}^{N} \left|\mathbf{g}_{k}(\mathbf{t}_{m,n})^{\mathrm{H}}\mathbf{b}_{k}\right|\right)^{2} \leq M^{2}N^{2}\left\|\mathbf{b}_{k}\right\|_{1}^{2},
\end{equation}
where the equality holds if the phase of the NFRV aligns with that of the PRV, i.e., $\mathbf{g}_{k}(\mathbf{t}_{m,n}) = \mathrm{e}^{\mathrm{j} \angle \mathbf{b}_{k}}$.

Substituting \eqref{eq_channel_power_one} into \eqref{eq_SNR_bound}, we obtain the upper bound on the SNR for user $k$ as
\begin{equation}\label{eq_SNR_bound2}
	\eta_{k}(\tilde{\mathbf{t}},\mathbf{w},\mathbf{p}) \leq \frac{MN\left\|\mathbf{b}_{k}\right\|_{1}^{2} p_{k}}{\sigma^{2}},~1 \leq k \leq K,
\end{equation}
which is the same as the upper bound on SINR shown in \eqref{eq_SINR_bound2}. Thus, the optimal power allocation for maximizing the minimum SNR can be obtained by solving the equations similar to \eqref{eq_power_allo_equ}. The corresponding upper bound on the minimum SNR is obtained as\footnote{Although the upper bound on the minimum SNR in \eqref{eq_minSNR_bound} is the same as that on the minimum SINR in \eqref{eq_minSINR_bound}, the digital and analog beamforming systems have totally different conditions for achieving the bound. Moreover, the analog beamforming system requires $K$ orthogonal frequency subcarriers to serve multiple users, while the digital beamforming system transmits signals over the same time-frequency resource block.}
\begin{equation}\label{eq_minSNR_bound}
	\bar{\eta} = \frac{P}{\sigma^{2}}\left(\sum \limits_{k=1}^{K} \frac{1}{MN\left\|\mathbf{b}_{k}\right\|_{1}^{2}}\right)^{-1}.
\end{equation}

Next, we provide a sufficient and necessary condition for the APV to achieve the upper bound under the special case that all the users have a single (LoS/NLoS) channel path and the size of each subarray is $N=1$, which is demonstrated in the following theorem. 
\begin{theorem}\label{Theo_ana}
	Under the condition of a single channel path for each user and $N=1$, the upper bound on the minimum SNR in \eqref{eq_minSNR_bound} is achieved if and only if the 3D position of each MA subarray belongs to the following set: 
	\begin{equation}\label{eq_opt_APV_ana}
		\begin{aligned}
			\mathbf{t}_{m} \in &\bigcap \limits_{2 \leq k \leq K} \Big{\{} \mathbf{t} \big{\arrowvert} \left\|\mathbf{t}-\mathbf{s}_{1}\right\|_{2} -\left\|\mathbf{t}-\mathbf{s}_{k}\right\|_{2} = \lambda(n_{k}+\phi_{k}),\\
			&~~~~~~~~~~~\forall n_{k} \in \mathbb{Z}, 0 \leq \phi_{k} < 1 \Big{\}} \triangleq \mathcal{T}.
		\end{aligned}
	\end{equation}
\end{theorem}
\begin{proof}
	See Appendix \ref{Appendix_ana}.
\end{proof}
Theorem \ref{Theo_ana} specifies the structure of the optimal APV solution for achieving the SNR upper bound under the condition of a single channel path for each user, where $\mathcal{T}$ is composed of the intersections of a series of hyperbolas determined by parameters $n_{k}$ and $\phi_{k}$. It can be used to obtain an optimal APV by checking the intersection of $\mathcal{T}$ and constraints \eqref{eq_problem_ABF_d} and \eqref{eq_problem_ABF_e}. For example, for the case of $K=2$, set $\mathcal{T}$ is simplified as
\begin{equation}
	\begin{aligned}
		\mathcal{T}= &\Big{\{} \mathbf{t} \big{\arrowvert} \left\|\mathbf{t}-\mathbf{s}_{1}\right\|_{2} -\left\|\mathbf{t}-\mathbf{s}_{2}\right\|_{2} = \lambda(n_{2}+\phi_{2}),\\
		&~~~~~~~~~~~~~~~~~~~~~~~~~\forall n_{2} \in \mathbb{Z}, 0 \leq \phi_{2} < 1 \Big{\}},
	\end{aligned}
\end{equation}
which is a series of hyperbolas w.r.t. focuses $\mathbf{s}_{1}$ and $\mathbf{s}_{2}$. As the number of users increases, the set $\mathcal{T}$ shrinks and it becomes more difficult to achieve the upper bound in \eqref{eq_minSNR_bound}. Next, suboptimal solutions for problem \eqref{eq_problem_ABF} are obtained by a general optimization algorithm.

\subsection{Optimization Algorithm}
For any given APV and ABFV, the optimal power allocation for maximizing the minimum SNR among all $K$ users is given by
\begin{equation}\label{eq_power_allo}
	p_{k} = P\left(\left|\mathbf{h}_{k}^{\mathrm{H}}(\tilde{\mathbf{t}}) \mathbf{w}\right|^{2} \sum \limits_{j=1}^{K} \frac{1}{\left|\mathbf{h}_{j}^{\mathrm{H}}(\tilde{\mathbf{t}}) \mathbf{w}\right|^{2}}\right)^{-1},~1 \leq k \leq K,
\end{equation}
which yields the same SNR for all $K$ users as
\begin{equation}\label{eq_SNR_min}
	\eta(\tilde{\mathbf{t}},\mathbf{w}) = \frac{P}{\sigma^{2}} \left( \sum \limits_{k=1}^{K} \frac{1}{\left|\mathbf{h}_{k}^{\mathrm{H}}(\tilde{\mathbf{t}}) \mathbf{w}\right|^{2}}\right)^{-1}.
\end{equation}

Due to the coupling between the APV and ABFV, we adopt the AO method. Specifically, denote the APV and ABFV obtained in the $(i-1)$-th iteration and $\tilde{\mathbf{t}}^{(i-1)}$ and $\mathbf{w}^{(i-1)}$, respectively. The gradient of $\eta(\tilde{\mathbf{t}},\mathbf{w})$ w.r.t. $\tilde{\mathbf{t}}$ is denoted as $\nabla_{\tilde{\mathbf{t}}} \eta(\tilde{\mathbf{t}},\mathbf{w})$, which is derived in Appendix \ref{Appendix_SNR_grad_APV}. Then, the APV is updated in the $i$-th iteration by the projected gradient ascent method as
\begin{equation}\label{eq_APV_update_ana}
	\tilde{\mathbf{t}}^{(i)} = \mathcal{B} \left\{ \tilde{\mathbf{t}}^{(i-1)} + \bar{\tau}^{(i)} \nabla_{\tilde{\mathbf{t}}} \eta(\tilde{\mathbf{t}}^{(i-1)},\mathbf{w}^{(i-1)}) \right\},
\end{equation}
where $\bar{\tau}^{(i)}$ is the step size obtained by backtracking line search until constraint \eqref{eq_problem_ABF_e} and the Armijo–Goldstein condition are both satisfied. 

Subsequently, we update the ABFV. Due to the constant-modulus constraint in \eqref{eq_problem_ABF_b}, it is not feasible to directly apply the gradient ascent method to ABFV optimization. To address this issue, we define the phase vector of the ABFV as $\boldsymbol{\varphi} = \angle \mathbf{w}$. The gradient of $\eta(\tilde{\mathbf{t}},\mathrm{e}^{\mathrm{j}\boldsymbol{\varphi}}/\sqrt{MN})$ w.r.t. $\boldsymbol{\varphi}$ is denoted as $\nabla_{\boldsymbol{\varphi}} \eta(\tilde{\mathbf{t}},\mathrm{e}^{\mathrm{j}\boldsymbol{\varphi}}/\sqrt{MN})$, which is derived in Appendix \ref{Appendix_SNR_grad_ABFV}. Then, the ABFV is updated in the $i$-th iteration by the gradient ascent method as
\begin{subequations}\label{eq_ABFV_update_ana}
	\begin{align}
		&\mathbf{w}^{(i)} = \frac{1}{\sqrt{MN}} \mathrm{e}^{\mathrm{j} \boldsymbol{\varphi}^{(i)}},\\
		&\boldsymbol{\varphi}^{(i)} = \boldsymbol{\varphi}^{(i-1)} + \hat{\tau}^{(i)} \nabla_{\boldsymbol{\varphi}} \eta(\tilde{\mathbf{t}}^{(i)},\mathrm{e}^{\mathrm{j}\boldsymbol{\varphi}^{(i-1)}}/\sqrt{MN}),
	\end{align}
\end{subequations}
where $\hat{\tau}^{(i)}$ is the step size obtained by backtracking line search until the Armijo–Goldstein condition is satisfied.

The overall algorithm for solving problem \eqref{eq_problem_ABF} is summarized in Algorithm \ref{alg_ABF}. In line 1, the APV is initialized by letting the subarrays uniformly distributed within the a subarea of the antenna moving region according to the user distribution, while the ABFV is initialized as $\mathbf{w}^{(0)}=\mathrm{e}^{\mathrm{j}\boldsymbol{\varphi}^{(0)}}/\sqrt{MN}$, $\boldsymbol{\varphi}^{(0)} = \angle (\sum_{k=1}^{K} \frac{\mathbf{h}_{k}^{\mathrm{H}}(\tilde{\mathbf{t}})}{\|\mathbf{h}_{k}^{\mathrm{H}}(\tilde{\mathbf{t}})\|_{2}})$. The APV and ABFV are alternatively optimized by gradient ascent in lines 4-17, where the algorithm terminates if the increment of the objective function is below a predefined threshold, $\epsilon_{2}$, or it reaches the maximum number of iterations, $I_{2}$. The convergence of Algorithm \ref{alg_ABF} is guaranteed because of the non-decreasing objective function over the iterations. The computational complexity of calculating the gradient in line 4 is $\mathcal{O}(MNLK)$, with $L=\max \limits_{1 \leq k \leq K}\{L_{k}\}$. The complexity of backtracking search in lines 7-10 is $\mathcal{O}(J_{2}MNLK)$, where $J_{2}$ is the maximum number of iterations for backtracking search. Besides, the computational complexity of calculating the gradient in line 11 is $\mathcal{O}(MNK)$. The complexity of backtracking search in lines 14-17 is $\mathcal{O}(J_{3}MNK)$, where $J_{3}$ is the maximum number of iterations for backtracking search. Thus, the total computational complexity of Algorithm \ref{alg_ABF} is $\mathcal{O}(I_{2}MNK(J_{2}L+J_{3}))$.

\begin{algorithm}[t]\small
	\caption{AO-based algorithm for problem \eqref{eq_problem_ABF}.}
	\label{alg_ABF}
	\begin{algorithmic}[1]
		\REQUIRE ~$M$, $N_{x}$, $N_{y}$, $K$, $P$, $\sigma^{2}$, $\lambda$, $A$, $d_{\min}$, $\{\mathbf{q}_{n}\}$, $\{L_{k}\}$,\\ ~~~~~~$\{\mathbf{s}_{k,\ell}\}$, $\{\mathbf{b}_{k}\}$, $\bar{\tau}$, $\bar{\xi}$, $\bar{\mu}$, $\hat{\tau}$, $\hat{\xi}$, $\hat{\mu}$, $\epsilon_{2}$, $I_{2}$.
		\ENSURE ~$\tilde{\mathbf{t}}^{\circ}$, $\mathbf{w}^{\circ}$, and $\mathbf{p}^{\circ}$. \\
		\STATE Initialize the APV $\tilde{\mathbf{t}}^{(0)}$ and the ABFV $\mathbf{w}^{(0)}$.
		\STATE Obtain $\eta(\tilde{\mathbf{t}}^{(0)},\mathbf{w}^{(0)})$ according to \eqref{eq_SNR_min}.
		\FOR   {$i=1:I_{2}$}
		\STATE Calculate the gradient according to \eqref{eq_SNR_grad_x} and \eqref{eq_SNR_grad_y}.
		\STATE Initialize step size $\bar{\tau}^{(i)}=\bar{\tau}$.
		\STATE Update the APV $\tilde{\mathbf{t}}^{(i)}$ according to \eqref{eq_APV_update_ana}.
		\WHILE {\eqref{eq_problem_ABF_e} or Armijo–Goldstein condition is not satisfied}
		\STATE Shrink the step size $\bar{\tau}^{(i)} \leftarrow \bar{\mu} \bar{\tau}^{(i)}$.
		\STATE Update the APV $\tilde{\mathbf{t}}^{(i)}$ according to \eqref{eq_APV_update_ana}.
		\ENDWHILE
		\STATE Calculate the gradient according to \eqref{eq_SNR_grad_w}.
		\STATE Initialize step size $\hat{\tau}^{(i)}=\hat{\tau}$.
		\STATE Update the ABFV $\mathbf{w}^{(i)}$ according to \eqref{eq_ABFV_update_ana}.
		\WHILE {Armijo–Goldstein condition is not satisfied}
		\STATE Shrink the step size $\hat{\tau}^{(i)} \leftarrow \hat{\mu} \hat{\tau}^{(i)}$.
		\STATE Update the ABFV $\mathbf{w}^{(i)}$ according to \eqref{eq_ABFV_update_ana}.
		\ENDWHILE
		\IF    {$\eta(\tilde{\mathbf{t}}^{(i)},\mathbf{w}^{(i)})-\eta(\tilde{\mathbf{t}}^{(i-1)},\mathbf{w}^{(i-1)})<\epsilon_{2}$}
		\STATE Break.
		\ENDIF
		\ENDFOR
		\STATE Set the APV and ABFV as $\tilde{\mathbf{t}}^{\circ}=\tilde{\mathbf{t}}^{(i)}$ and $\mathbf{w}^{\circ}=\mathbf{w}^{(i)}$.
		\STATE Calculate the power allocation vector $\mathbf{p}^{\circ}$ according to \eqref{eq_power_allo}.
		\RETURN $\tilde{\mathbf{t}}^{\circ}$, $\mathbf{w}^{\circ}$, and $\mathbf{p}^{\circ}$.
	\end{algorithmic}
\end{algorithm}

\subsection{Statistical CSI Based Antenna Position Optimization}
To reduce antenna movement overhead, we propose to design the APV based on the statistical CSI over a long time period, while the ABFV is optimized based on instantaneous CSI. The corresponding two-timescale optimization problem for maximizing the ergodic SNR of $K$ users is thus given by
\begin{subequations}\label{eq_problem_ABF_ergodic}
	\begin{align}
		&\mathop{\max}\limits_{\tilde{\mathbf{t}}}~ \mathbb{E}\left\{\mathop{\max}\limits_{\mathbf{w},\mathbf{p}}~
		\min \limits_{1 \leq k \leq K} ~ \eta_{k}(\tilde{\mathbf{t}},\mathbf{w},\mathbf{p}) \right\} \label{eq_problem_ABF_ergodic_a}\\
		&~~\mathrm{s.t.}~~  \eqref{eq_problem_ABF_b},~\eqref{eq_problem_ABF_c},~\eqref{eq_problem_ABF_d},~\eqref{eq_problem_ABF_e},
	\end{align}
\end{subequations}
where the expectation in \eqref{eq_problem_DBF_ergodic_a} is operated over the random instantaneous channels and similar to \eqref{eq_SINR_ergodic}, it is approximated by Monte Carlo simulations as
\begin{equation}\label{eq_SNR_ergodic}
	{\small
		\begin{aligned}
			\mathbb{E}\left\{\mathop{\max}\limits_{\mathbf{w},\mathbf{p}}
			\min \limits_{1 \leq k \leq K}  \eta_{k}(\tilde{\mathbf{t}},\mathbf{w},\mathbf{p}) \right\} \approx \frac{1}{Q_{2}} \sum \limits_{q=1}^{Q_{2}} \mathop{\max}\limits_{\mathbf{w},\mathbf{p}}
			\min \limits_{1 \leq k \leq K}  \eta_{k}^{q}(\tilde{\mathbf{t}},\mathbf{w},\mathbf{p}),
		\end{aligned}
	}
\end{equation}
where $Q_{2}$ is the total number of Monte Carlo simulations and $\eta_{k}^{q}$ denotes the SNR for user $k$ under the $q$-th, $1 \leq q \leq Q_{2}$, independent channel realization. For each independent channel realization, we can also adopt the AO-based approach in Algorithm \ref{alg_ABF}, whereas the gradient w.r.t. APV should be averaged over all the $Q_{2}$ independent channel realizations based on the given statistical CSI knowledge.

\section{Performance Evaluation}
In this section, we provide the simulation results of the considered MA array-enabled near-field communication systems. The simulation setup and benchmark schemes are first illustrated, and then the numerical results under digital beamforming and analog beamforming architectures are presented. 

\subsection{Simulation Setup and Benchmark Schemes}
In the simulation, the height of the BS's MA array is set as $15$ m. The total number of antennas at the BS is set as $MN=64$. Unless otherwise stated, each subarray consists of one antenna for simplicity, i.e. $N=1$. For the subarray with more antennas, it is set as an UPA of size $N_{x} \times N_{y}$ with inter-antenna spacing $\frac{\lambda}{2}$. The antenna moving region is set as a square area of size $A \times A$, with $A=100\lambda$. The minimum distance between MA subarrays is set as $d_{\min}=\frac{\lambda}{2} \times \max\{N_{x},N_{y}\}$. The total number of users is set as $K=32$. They are each equipped with a single FPA and randomly distributed on the ground with zero altitude. The carrier frequency is set as $f_{\mathrm{c}}=30$ GHz. An LoS path is assumed to exist between the BS and each user, with the amplitude of the path response coefficient given by $[\mathbf{b}_{k}]_{0}=\frac{\lambda}{4\pi\|\mathbf{s}_{k,0}\|_{2}}$ \cite{liu2023near}. Due to the significant propagation/penetration loss of millimeter-wave signals, it is assumed that NLoS paths experience severe attenuation and thus have negligible power compared to the LoS path. The maximum transmit power of the BS is $P=20$ dBm and the noise power is $\sigma^{2}=-80$ dBm. For both Algorithms \ref{alg_DBF} and \ref{alg_ABF}, the maximum number of iterations is set as $I_{1}=I_{2}=300$. The threshold for terminating the iteration is set as $\epsilon_{1}=\epsilon_{2}=10^{-5}$. The initial step size for backtracking line search is $\tau=\bar{\tau}=\hat{\tau}=10\lambda$. The factor of shrinking the step size is set as $\mu = \bar{\mu} = \hat{\mu}=0.5$ and the increment speed factor is set as $\xi = \bar{\xi} = \hat{\xi}=0.1$. All the results in this section are the average performance over $Q_{1}=Q_{2}=500$ Monte Carlo simulations of random user locations and channel realizations.

In addition to the proposed MA solution based on instantaneous/statistical CSI, we consider the following FPA-based benchmark schemes for performance comparison.
\begin{itemize}
	\item Upper bound: The upper bound on the minimum SINR for MA-aided digital beamforming systems is given by \eqref{eq_minSINR_bound} and that on the minimum SNR for MA-aided analog beamforming systems is given by \eqref{eq_minSNR_bound};
	\item Dense UPA: The antennas are densely placed to form an UPA of size $8 \times 8$, with inter-antenna spacing $\lambda/2$;
	\item Sparse UPA: The antennas are sparsely placed to form an UPA of size $8 \times 8$, with inter-antenna spacing $A/8$;
	\item Horizontal sparse UPA: The UPA of size $8 \times 8$ is sparse in the horizontal dimension with inter-antenna spacing $A/8$, while dense in the vertical dimension with inter-antenna spacing $\lambda/2$;
	\item Vertical sparse UPA: The UPA of size $8 \times 8$ is sparse in the vertical dimension with inter-antenna spacing $A/8$, while dense in the horizontal dimension with inter-antenna spacing $\lambda/2$;
	\item Horizontal sparse ULA: The antennas are horizontally placed to form a uniform linear array (ULA), with inter-antenna spacing $A/64$;
	\item Vertical sparse ULA: The antennas are vertically placed to form an ULA, with inter-antenna spacing $A/64$;
\end{itemize}
Note that for fair comparison, all the above benchmarks adopt the same number of antennas as that of MA arrays, i.e., $64$. In addition, the DBFM and ABFV are optimized in a similar way to that for MA systems.

\subsection{Numerical Results of Digital Beamforming System}
The main design objective for digital beamforming systems is to enhance the spatial multiplexing performance of multiuser communications. Thus, we consider that the users are randomly distributed within a semicircular region on the ground centered at the BS, where the horizontal distance between the BS and users ranges from $5$ to $50$ m. 
\begin{figure}[t]
	\begin{center}
		\includegraphics[width=\figwidth cm]{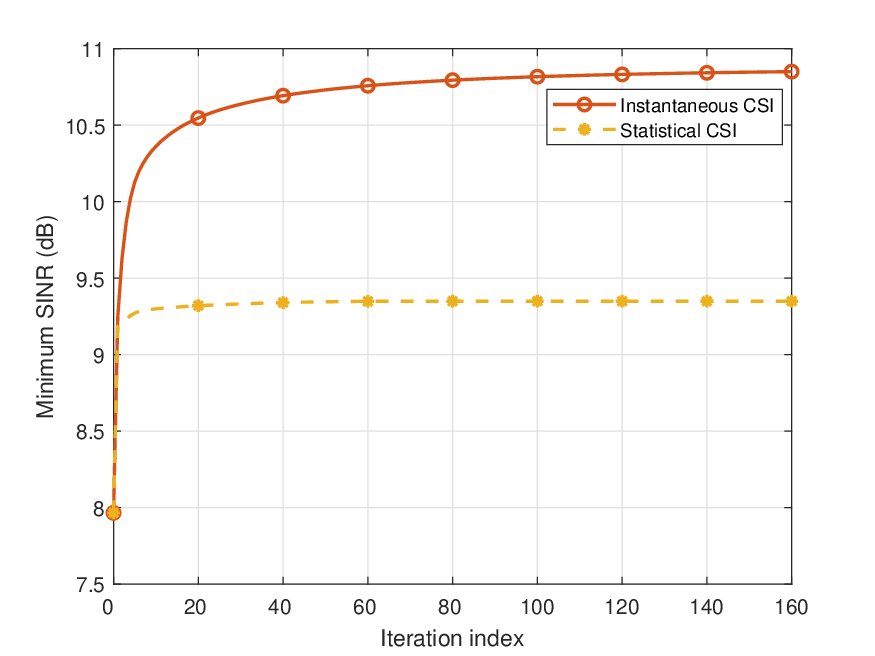}
		\caption{Convergence evaluation of Algorithm \ref{alg_DBF}.}
		\label{Fig_digital_Iteration}
	\end{center}
\end{figure}

We first evaluate in Fig. \ref{Fig_digital_Iteration} the convergence performance of the proposed Algorithm \ref{alg_DBF}. It is observed that for MA systems based on instantaneous CSI or statistical CSI, the algorithm converges after 120 iterations and 40 iterations, respectively, which demonstrates the efficacy of the proposed algorithms. In addition, due to the common APV adopted for all independent channel realizations, the statistical CSI-based MA scheme suffers from a performance loss in the minimum SINR among multiple users. Nonetheless, since the statistical CSI-based MA scheme does not require frequent antenna movement, it can achieve a good trade-off between the communication performance and system overhead.

\begin{figure}[t]
	\begin{center}
		\includegraphics[width=\figwidth cm]{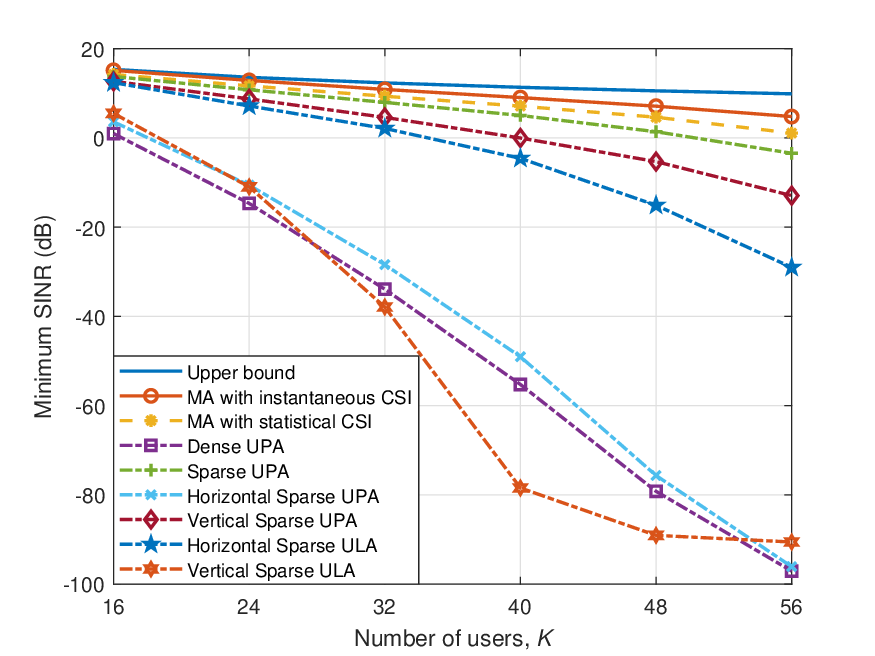}
		\caption{Performance comparison of the minimum SINR achieved by the proposed and benchmark schemes versus the number of users.}
		\label{Fig_digital_K}
	\end{center}
\end{figure}
Fig. \ref{Fig_digital_K} compares the SINR performance of the proposed MA and benchmark schemes with varying number of users. The proposed MA schemes always outperform other benchmark schemes with FPAs because the antenna position optimization can significantly decrease the correlation of users' channel vectors. When the number of users is small, the MA schemes can even approach the performance upper bound on the minimum SINR. As the number of users increases, the SINR performance for all schemes decreases because the BS has to allocate the transmit power to more users and reducing their channel correlation via antenna precoding becomes more difficult. In addition, some useful insights can be observed from the benchmark schemes. Comparing the dense UPA with the sparse UPA, we can find that increasing the array sparsity can decrease the users' channel correlation and thus improve the multiplexing performance. Moreover, the vertical sparse UPA can achieve a performance much better than that of the horizontal sparse UPA. This is because the users are uniformly distributed over the horizontal dimension yet non-uniformly distributed over the vertical dimension. Thus, the horizontal sparsity of the UPA cannot contribute much to decreasing the channel correlation of users. Moreover, the horizontal sparse ULA can achieve dominant near-field effects in the horizontal plane, which helps decrease the users' channel correlation. In comparison, the vertical sparse ULA cannot distinguish with the channel vectors of users with the same distance to the BS on the same horizontal plane, which results in lower SINR as compared to the horizontal sparse ULA. 

\begin{figure}[t]
	\begin{center}
		\includegraphics[width=\figwidth cm]{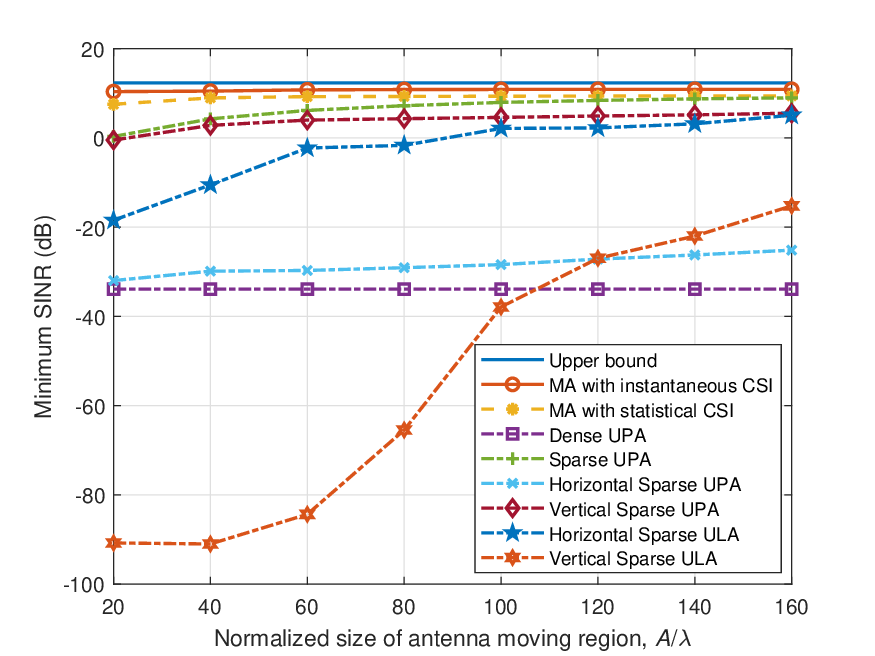}
		\caption{Comparison of the minimum SINR achieved by the proposed and benchmark schemes versus the size of the antenna moving region.}
		\label{Fig_digital_A}
	\end{center}
\end{figure}
In Fig. \ref{Fig_digital_A}, we compare the SINR performance achieved by the proposed and benchmark schemes with varying size of the antenna moving region. The results validate again that the proposed MA schemes can significantly outperform all FPA-based benchmark schemes. It is worth noting that for a small region size, $A=40 \lambda$, the SINR performance of the MA schemes is even comparable to that for $A=160\lambda$. However, the sparse UPA/ULA schemes require larger region sizes for achieving the same SINR performance. It indicates that MA systems can efficiently exploit the DoFs in the spatial domain via antenna position optimization for decreasing the channel correlation among multiple users.

\begin{figure}[t]
	\begin{center}
		\includegraphics[width=\figwidth cm]{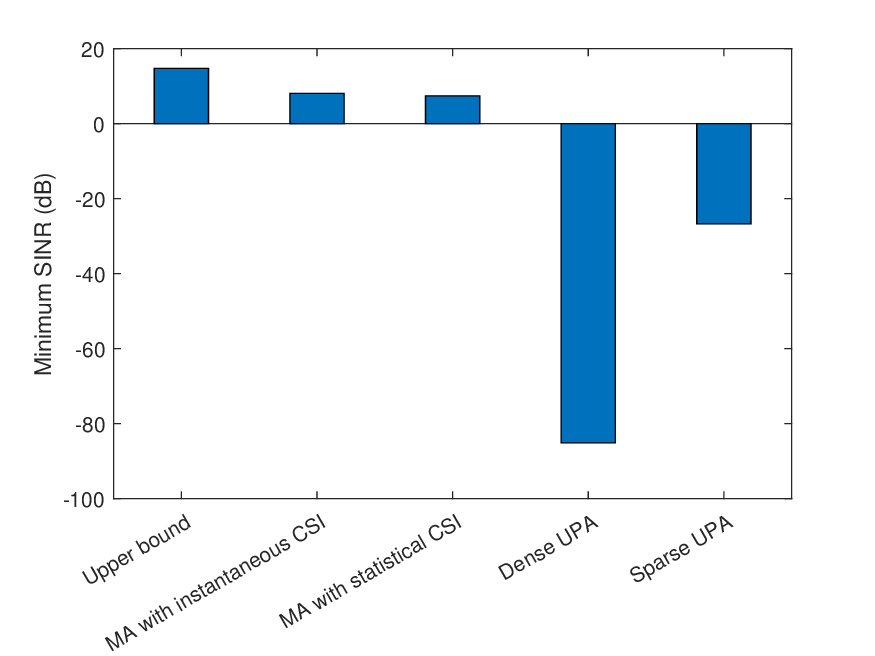}
		\caption{Comparison of the minimum SINR achieved by the proposed and benchmark schemes under nonuniform user distribution.}
		\label{Fig_digital_SINR}
	\end{center}
\end{figure}
\begin{figure}[t]
	\begin{center}
		\includegraphics[width=\figwidth cm]{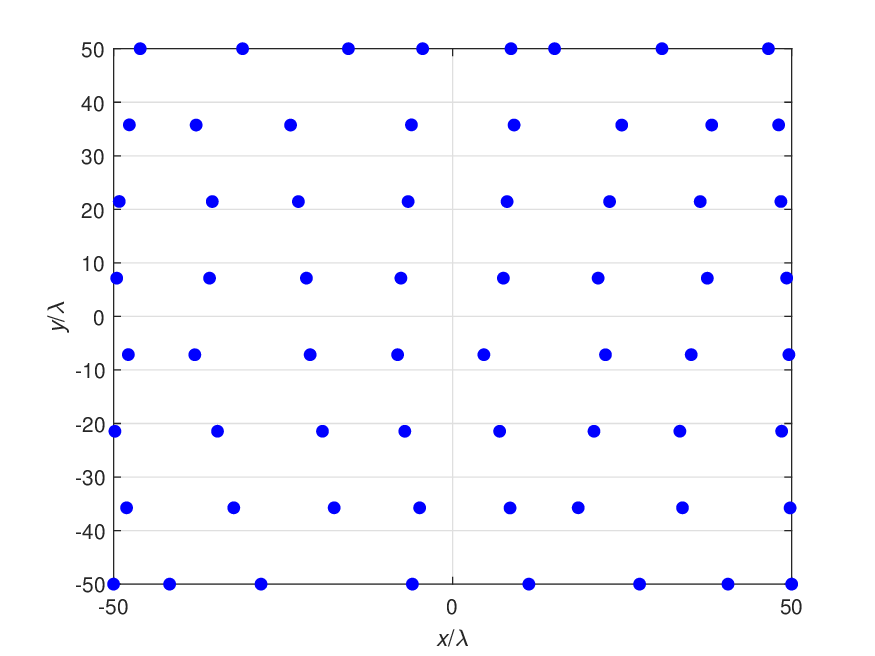}
		\caption{Geometry of the optimized MA array under the digital beamforming architecture.}
		\label{Fig_digital_Geometry}
	\end{center}
\end{figure}
\begin{figure*}[t]
	\centering
	\subfigure[MA array]{\includegraphics[width=5.5 cm]{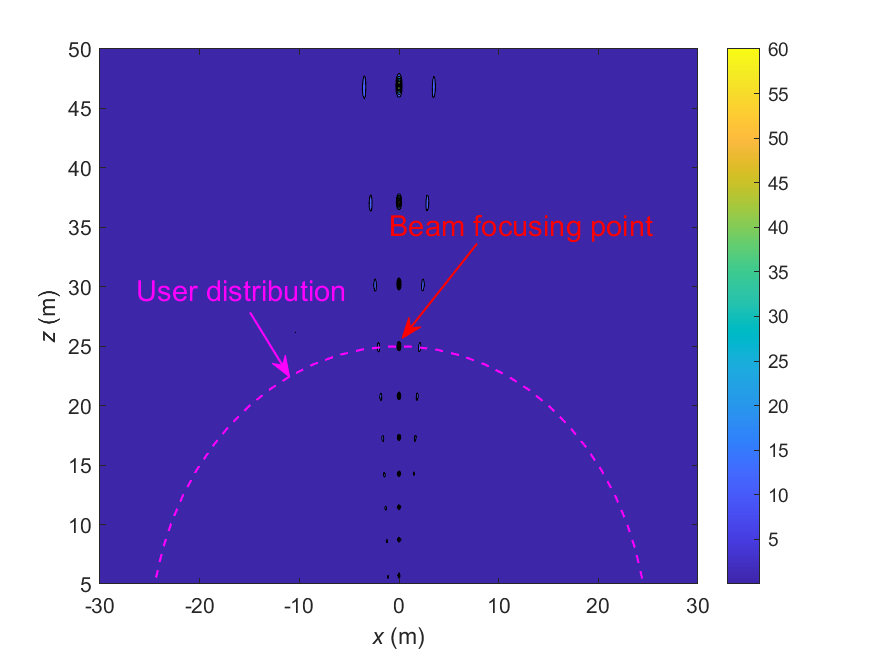} \label{Fig_digital_Beam_MA}}
	\subfigure[Dense UPA]{\includegraphics[width=5.5 cm]{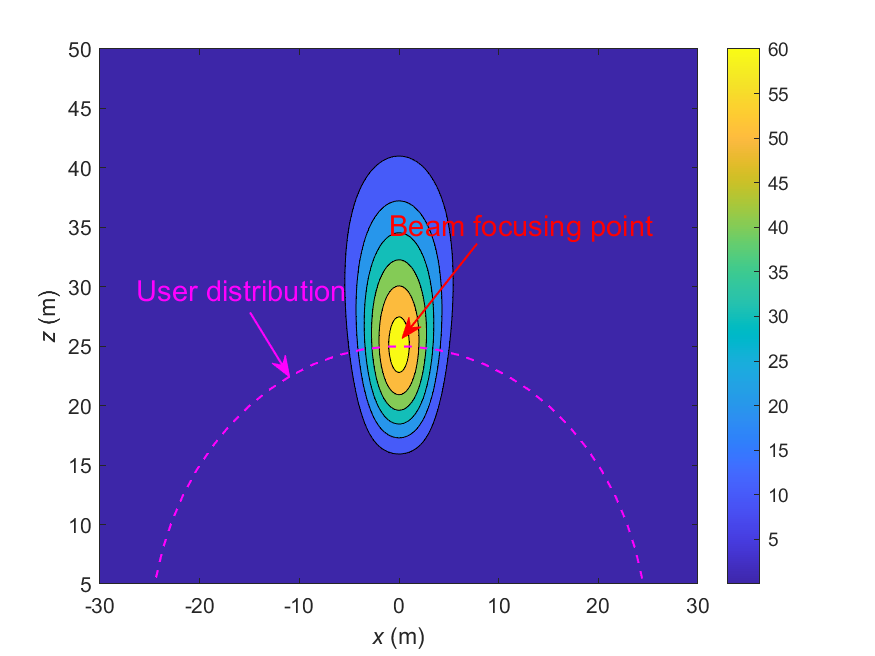} \label{Fig_digital_Beam_dense}}
	\subfigure[Sparse UPA]{\includegraphics[width=5.5 cm]{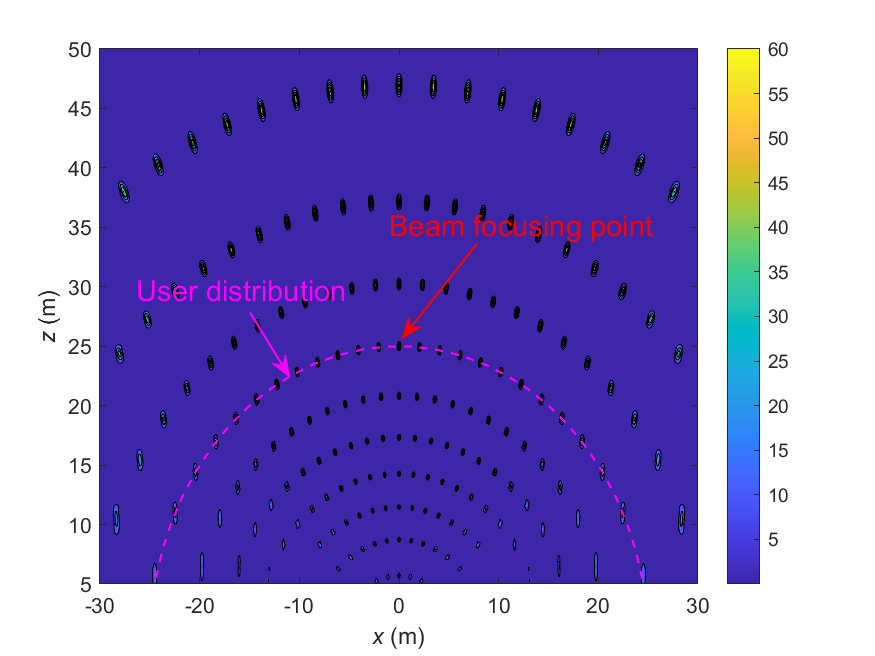} \label{Fig_digital_Beam_sparse}}
	\caption{Beam pattern of different antenna arrays under the digital beamforming architecture.}
	\label{Fig_digital_Beam}
\end{figure*}
To further demonstrate the advantage of MA systems, we also consider the case of non-uniform user distribution. Specifically, the users are randomly distributed within a narrow semicircular region, with the horizontal distance to the BS ranging from $24.9$ to $25.1$ m. The corresponding minimum SINR performance is shown in Fig. \ref{Fig_digital_SINR}. It is observed that the proposed MA scheme can achieve much higher performance gains over both dense UPA and sparse UPA schemes in this case. Moreover, we show in Fig. \ref{Fig_digital_Geometry} the optimized MA array geometry, i.e., the positions of antennas, which are obtained based on the statistical CSI (i.e., the non-uniform user distribution). As can be observed, the MAs are sparsely spaced in a non-uniform manner, which can help decrease the average channel correlation of ground users distributed in the considered region. The beam pattern corresponding to this optimized MA array geometry is shown in Fig. \ref{Fig_digital_Beam_MA}. Specifically, we calculate the near-field array response vector, $[\mathrm{e}^{\mathrm{j}\frac{2\pi}{\lambda}\|\mathbf{t}_{1,1}-\mathbf{s}\|_{2}},\dots,
\mathrm{e}^{\mathrm{j}\frac{2\pi}{\lambda}\|\mathbf{t}_{M,N}-\mathbf{s}\|_{2}}]^{\mathrm{T}}$, over all locations $\mathbf{s}$ within a rectangular area on the ground. The beam is focused on location ($x=0$ m, $z=25$ m), which is realized by setting the beamforming vector as the normalized near-field array response vector over this beam focusing point. Then, the beamforming gain is calculated as the power of the inner product between the near-field array response vector and the beamforming vector. As shown in Fig. \ref{Fig_digital_Beam_MA}, the main lobe of the beam is very narrow and only a small number of sidelobes exist in the considered user region. This indicates that beam that is focused on the position of a user causes little interference leakage to users located at other positions. Thus, the multiuser communication performance can be significantly improved by MA arrays. In comparison, the beam of the dense UPA shown in Fig. \ref{Fig_digital_Beam_dense} has a wide main lobe, while that of the sparse UPA exhibits many strong sidelobes shown in Fig. \ref{Fig_digital_Beam_sparse}, both of which lead to a high channel correlation for users distributed within this region. If the ZF beamformer is adopted at the BS to eliminate the multiuser interference, the effective channel gain of desired user signals will be severely reduced in the benchmark FPA systems. As such, the MA system can significantly outperform FPA systems by flexibly reconfiguring the array geometry to adapt to the user distribution/statistical CSI.

\subsection{Numerical Results of Analog Beamforming System}
Next, we consider the analog beamforming system which aims to increase the beam coverage performance for multiple users. Thus, we consider that the users are randomly distributed within two circular hotspot regions centered at ($x=-25$ m, $z=40$ m) and ($x=25$ m, $z=40$ m), respectively, with the radius of $2.5$ m.

\begin{figure}[t]
	\begin{center}
		\includegraphics[width=\figwidth cm]{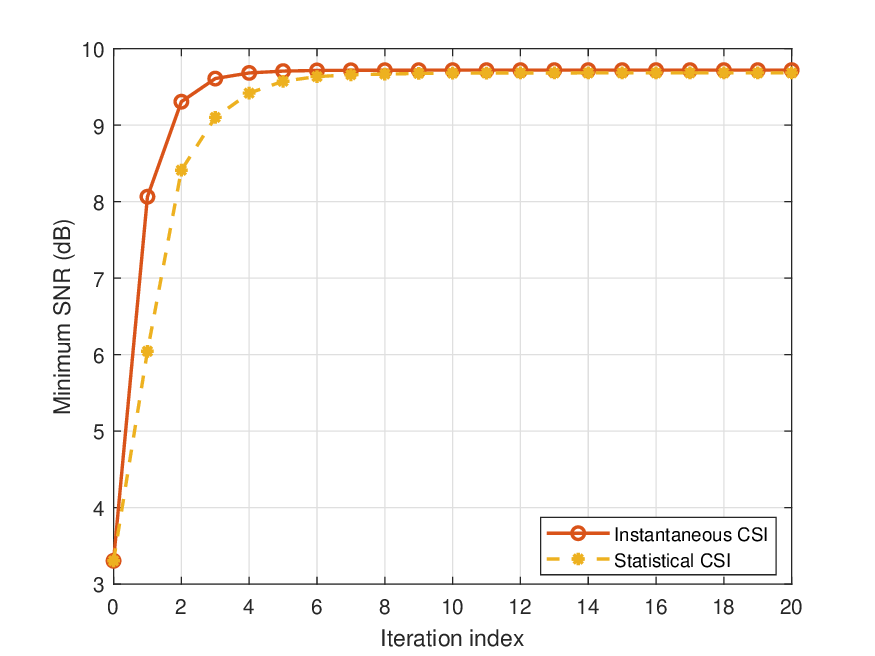}
		\caption{Convergence evaluation of Algorithm \ref{alg_ABF}.}
		\label{Fig_analog_Iteration}
	\end{center}
\end{figure}
Fig. \ref{Fig_analog_Iteration} evaluates the convergence performance of the proposed Algorithm \ref{alg_ABF}. It is observed that the algorithm achieves a rapid convergence within 8 iterations for MA systems based on both instantaneous CSI and statistical CSI. Moreover, the performance of the statistical CSI-based MA scheme is almost the same as that of the instantaneous CSI-based MA scheme in terms of minimum SNR among multiple users. This is because the users are all distributed within the two hotspot regions. The optimized MA array based on statistical CSI can yield high correlation among the channel vectors of all users in both regions, which thus circumvents the frequent antenna movement.

\begin{figure}[t]
	\begin{center}
		\includegraphics[width=\figwidth cm]{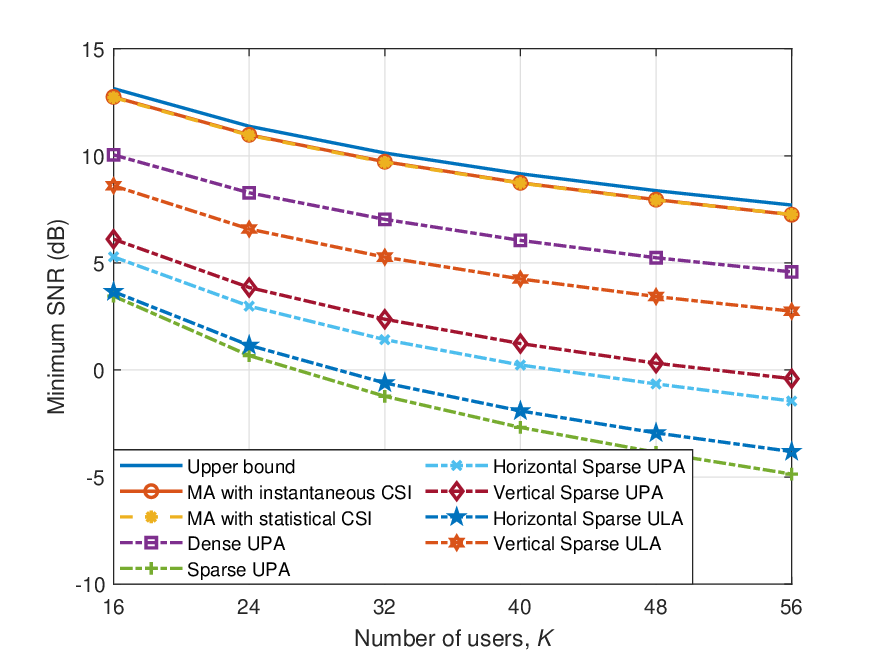}
		\caption{Performance comparison of the minimum SNR achieved by the proposed and benchmark schemes versus the number of users.}
		\label{Fig_analog_K}
	\end{center}
\end{figure}
In Fig. \ref{Fig_analog_K}, we compare the SNR performance achieved by the proposed and benchmark schemes with varying number of users. We can observe that the proposed MA schemes significantly outperform other benchmark schemes based on FPAs. In contrast to the digital beamforming system, the analog beamforming system requires a high correlation of the channel vectors for multiple users such that they can achieve high beamforming gains at the same time under the considered OFDMA-based multiple access. Thus, the dense UPA scheme performs better than the sparse UPA scheme in this case. Moreover, as the number of users increases, the minimum SNR decreases because the BS has to allocate the transmit power to more users. Nevertheless, the proposed MA schemes can achieve a performance close to the derived upper bound due to the flexibility in antenna position optimization.

\begin{figure}[t]
	\begin{center}
		\includegraphics[width=\figwidth cm]{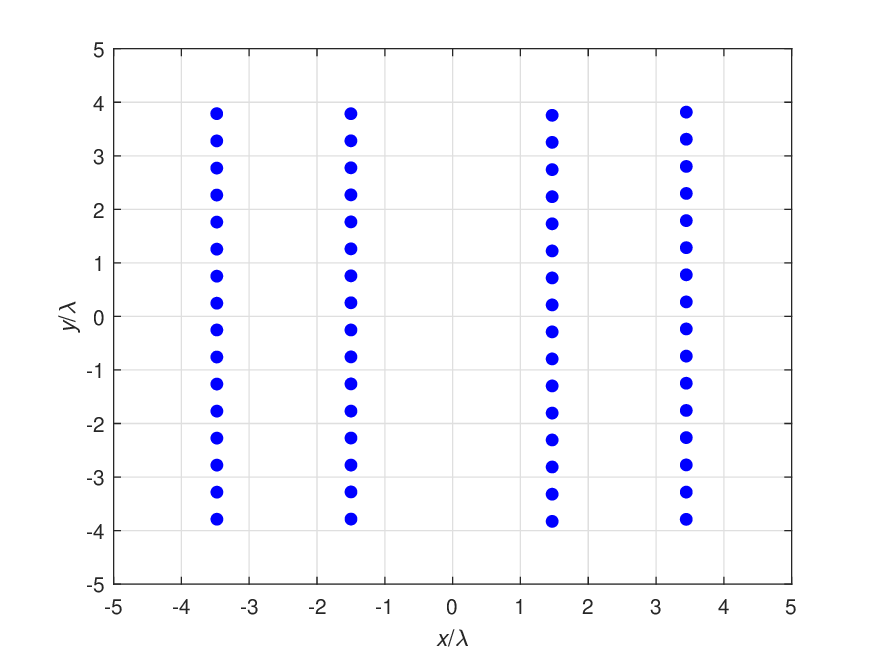}
		\caption{Geometry of the optimized MA array under the analog beamforming architecture.}
		\label{Fig_analog_Geometry}
	\end{center}
\end{figure}
\begin{figure}[t]
	\begin{center}
		\includegraphics[width=\figwidth cm]{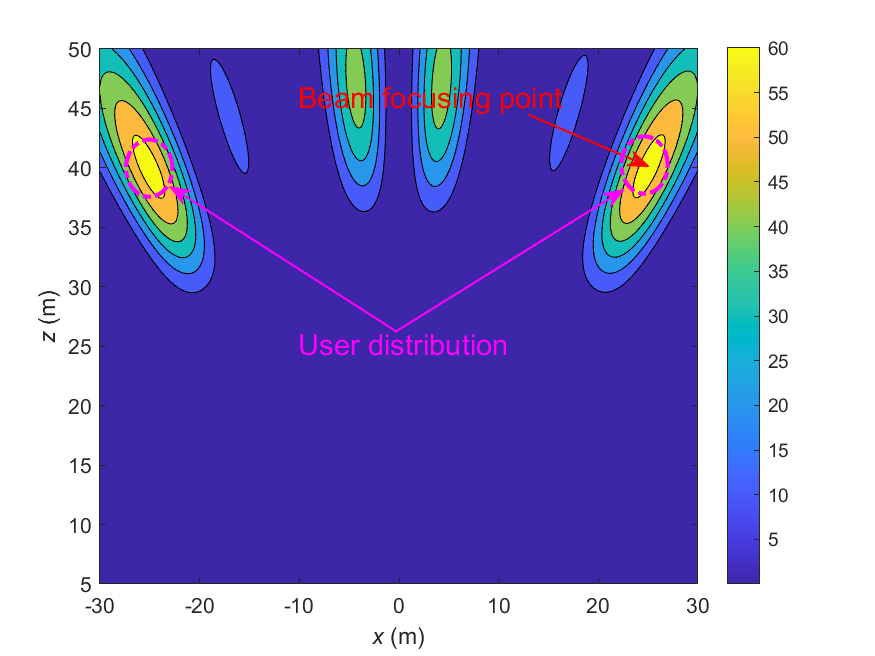}
		\caption{Beam pattern of the optimized MA array under the analog beamforming architecture.}
		\label{Fig_analog_Beam}
	\end{center}
\end{figure}
Next, we show the MA array geometry optimized based on statistical CSI and the corresponding beam pattern in Figs. \ref{Fig_analog_Geometry} and \ref{Fig_analog_Beam}, respectively. In particular, the beam focusing point is set as the center of the hotspot area at the right-hand side. It is observed that the antennas are uniformly distributed in the vertical dimension with a small spacing, which can guarantee a wide beam over this direction. Besides, the antennas are sparsely spaced in the horizontal dimension, which can guarantee a high grating lobe over the other hotspot area at the left-hand side shown in Fig. \ref{Fig_analog_Beam}. By utilizing the optimized MA array, the channel vectors of all users within the two hotspot areas are highly correlated such that the ABFV can yield a beamforming gain close to the upper bound, $MN=64$, for all users.

\begin{figure}[t]
	\begin{center}
		\includegraphics[width=\figwidth cm]{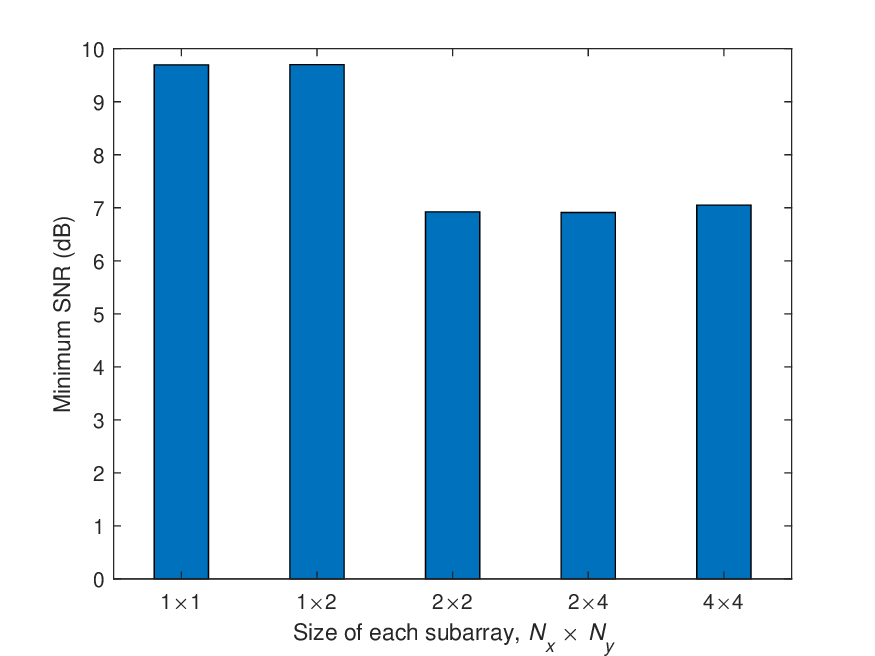}
		\caption{Minimum SNR of the proposed MA scheme based on statistical CSI versus the size of each subarray.}
		\label{Fig_analog_NxNy}
	\end{center}
\end{figure}
Finally, Fig. \ref{Fig_analog_NxNy} shows the performance of the proposed MA scheme based on statistical CSI under different sizes of each subarray. Since the total number of antennas is fixed as $MN=64$, as the size of each subarray increases, the number of subarrays, $M$, decreases. This results in reduced DoFs in antenna position optimization such that the minimum SNR performance of the consider MA-aided analog beamforming system deteriorates. Thus, in practice, an appropriate size of the subarray should be chosen to achieve a desired trade-off between the system performance and hardware cost.

\section{Conclusions}
In this paper, we investigated the MA-enabled near-field communications, where the BS is equipped with multiple subarrays movable in a large region to serve multiple users each equipped with a single FPA. First, we extended the field response channel model customized for MA systems from the far field to the near-field propagation condition, where the spherical wave model was adopted to characterize the spatial variation of wireless channels w.r.t. each MA's position at the BS. Then, we characterized the performance of MA-aided multiuser communication systems under both digital beamforming and analog beamforming architectures. Specifically, for MA systems with digital beamforming, an upper bound on the minimum SINR among multiple users was derived in closed form, following which a low-complexity ZF-based algorithm was developed to jointly optimize the APV and DBFM at the BS for approaching this upper bound. Besides, for MA systems with analog beamforming, an upper bound on the minimum receive SNR among multiple users was derived in closed form. Thereafter, an AO-based algorithm was developed to iteratively optimize the APV, ABFV, and power allocation until convergence. For both architectures, we further extended the designs of MA systems to the case based on statistical CSI, where the APV is updated over a long time period to reduce the antenna movement overhead. Simulation results showed that our proposed algorithms can approach the performance upper bound on the minimum SINR/SNR under the digital/analog beamforming architecture and the MA systems significantly outperform other benchmark schemes based on dense or spare arrays with FPAs. Moreover, our proposed strategy with statistical CSI-based antenna position optimization can achieve a performance comparable to that based on instantaneous CSI, which provides a viable solution for future MA-enabled BSs.

\appendices

\section{Proof of Theorem \ref{Theo_dig}}\label{Appendix_dig}
Under the condition of a single channel path for each user and $N=1$ for each subarray, the channel vector between the BS and user $k$ in \eqref{eq_channel_vec} is simplified as $\mathbf{h}_{k}(\tilde{\mathbf{t}}) = b_{k}[\mathrm{e}^{\mathrm{j}\frac{2\pi}{\lambda} \|\mathbf{t}_{1}-\mathbf{s}_{k}\|_{2}},\mathrm{e}^{\mathrm{j}\frac{2\pi}{\lambda} \|\mathbf{t}_{2}-\mathbf{s}_{k}\|_{2}},\dots,\mathrm{e}^{\mathrm{j}\frac{2\pi}{\lambda} \|\mathbf{t}_{M}-\mathbf{s}_{k}\|_{2}}]^{\mathrm{T}}$. Then, the equality in \eqref{eq_channel_power} always holds, i.e., $\left\|\mathbf{h}_{k}(\tilde{\mathbf{t}})\right\|_{2}^{2} = MN|b_{k}|^{2}$. To guarantee the orthogonality between the channel vectors for different users, i.e., the equality in \eqref{eq_SINR_bound} to hold, the sufficient as well as necessary condition can be derived as
\begin{equation}\label{eq_orthogonal_condi}
	\begin{aligned}
		&\mathbf{h}_{k}^{\mathrm{H}}(\tilde{\mathbf{t}})\mathbf{h}_{\hat{k}}(\tilde{\mathbf{t}}) = 0,~1 \leq k, \hat{k} \leq K\\
		\Leftrightarrow&  \sum \limits_{m=1}^{M} \mathrm{e}^{\mathrm{j}\frac{2\pi}{\lambda} (\|\mathbf{t}_{m} -\mathbf{s}_{\hat{k}}\|_{2}  - \|\mathbf{t}_{m} -\mathbf{s}_{k}\|_{2} )} = 0\\
		\Leftrightarrow&  \|\mathbf{t}_{m} -\mathbf{s}_{\hat{k}}\|_{2}  - \|\mathbf{t}_{m} -\mathbf{s}_{k}\|_{2} = \lambda(n_{k,\hat{k},m}+\phi_{k,\hat{k},m}),  \\
		&\sum \limits_{m=1}^{M} \mathrm{e}^{\mathrm{j} 2\pi \phi_{k,\hat{k},m}} = 0.
	\end{aligned}
\end{equation}
This thus completes the proof.

\section{Derivation of Gradient $\nabla_{\tilde{\mathbf{t}}} \gamma(\tilde{\mathbf{t}})$}\label{Appendix_SINR_grad}
The gradient $\nabla_{\tilde{\mathbf{t}}} \gamma(\tilde{\mathbf{t}})$ is composed of the partial derivative of $\gamma(\tilde{\mathbf{t}})$ w.r.t. $x_{m}$ and $y_{m}$, $1 \leq m \leq M$. Specifically, we have
\begin{equation}\label{eq_SINR_grad_x}{\small
	\begin{aligned}
		&\frac{\partial \gamma(\tilde{\mathbf{t}})}{\partial x_{m}} = -\frac{P}{\mathrm{tr}\{\mathbf{Z}(\tilde{\mathbf{t}})^{-1}\}^{2} \sigma^{2}} \mathrm{tr}\left\{\frac{\partial \mathbf{Z}(\tilde{\mathbf{t}})^{-1}}{\partial x_{m}}\right\}\\
		&=\frac{P}{\mathrm{tr}\{\mathbf{Z}(\tilde{\mathbf{t}})^{-1}\}^{2} \sigma^{2}} \mathrm{tr}\left\{(\mathbf{Z}(\tilde{\mathbf{t}})^{*})^{-1} \frac{\partial \mathbf{Z}(\tilde{\mathbf{t}})}{\partial x_{m}} (\mathbf{Z}(\tilde{\mathbf{t}})^{*})^{-1}\right\},
	\end{aligned}}
\end{equation}
with the entry of $\frac{\partial \mathbf{Z}(\tilde{\mathbf{t}})}{\partial x_{m}}$ in row $i$ and column $j$ given by
\begin{equation*}{\small
	\begin{aligned}
		&\left[\frac{\partial \mathbf{Z}(\tilde{\mathbf{t}})}{\partial x_{m}}\right]_{i,j} = \mathbf{b}_{i}^{\mathrm{H}} \frac{\partial \mathbf{G}_{i}(\tilde{\mathbf{t}})\mathbf{G}_{j}(\tilde{\mathbf{t}})^{\mathrm{H}}}{\partial x_{m}} \mathbf{b}_{j}, ~1 \leq i, j \leq K,\\
		&\left[\frac{\partial \mathbf{G}_{i}(\tilde{\mathbf{t}})\mathbf{G}_{j}(\tilde{\mathbf{t}})^{\mathrm{H}}}{\partial x_{m}}\right]_{l,\ell}
		= \frac{\partial}{\partial x_{m}} \sum \limits_{n=1}^{N} \mathrm{e}^{\mathrm{j} \frac{2\pi}{\lambda} \left(\|\mathbf{t}_{m,n}-\mathbf{s}_{i,l}\|_{2} - \|\mathbf{t}_{m,n}-\mathbf{s}_{j,\ell}\|_{2}\right)}\\
		&=\sum \limits_{n=1}^{N} \mathrm{j} \frac{2\pi}{\lambda} \left(\frac{x_{m}+[\mathbf{q}_{n}]_{1}-[\mathbf{s}_{i,l}]_{1}}{\|\mathbf{t}_{m}+\mathbf{q}_{n}-\mathbf{s}_{i,l}\|_{2}} 
		- \frac{x_{m}+[\mathbf{q}_{n}]_{1}-[\mathbf{s}_{j,\ell}]_{1}}{\|\mathbf{t}_{m}+\mathbf{q}_{n}-\mathbf{s}_{j,\ell}\|_{2}} \right)\\
		& ~~\times \mathrm{e}^{\mathrm{j} \frac{2\pi}{\lambda} \left(\|\mathbf{t}_{m,n}-\mathbf{s}_{i,l}\|_{2} - \|\mathbf{t}_{m,n}-\mathbf{s}_{j,\ell}\|_{2}\right)}, 1 \leq l \leq L_{i}, ~1 \leq \ell \leq L_{j}.
	\end{aligned}}
\end{equation*}
Similarly, we have 
\begin{equation}\label{eq_SINR_grad_y}{\small
	\begin{aligned}
		&\frac{\partial \gamma(\tilde{\mathbf{t}})}{\partial y_{m}} =\frac{P}{\mathrm{tr}\{\mathbf{Z}(\tilde{\mathbf{t}})^{-1}\}^{2} \sigma^{2}} \mathrm{tr}\left\{(\mathbf{Z}(\tilde{\mathbf{t}})^{*})^{-1} \frac{\partial \mathbf{Z}(\tilde{\mathbf{t}})}{\partial y_{m}} (\mathbf{Z}(\tilde{\mathbf{t}})^{*})^{-1}\right\},
	\end{aligned}}
\end{equation}
with the entry of $\frac{\partial \mathbf{Z}(\tilde{\mathbf{t}})}{\partial y_{m}}$ in row $i$ and column $j$ given by
\begin{equation*}{\small
		\begin{aligned}
			&\left[\frac{\partial \mathbf{Z}(\tilde{\mathbf{t}})}{\partial y_{m}}\right]_{i,j} = \mathbf{b}_{i}^{\mathrm{H}} \frac{\partial \mathbf{G}_{i}(\tilde{\mathbf{t}})\mathbf{G}_{j}(\tilde{\mathbf{t}})^{\mathrm{H}}}{\partial y_{m}} \mathbf{b}_{j}, ~1 \leq i, j \leq K,\\
			&\left[\frac{\partial \mathbf{G}_{i}(\tilde{\mathbf{t}})\mathbf{G}_{j}(\tilde{\mathbf{t}})^{\mathrm{H}}}{\partial y_{m}}\right]_{l,\ell}
			= \frac{\partial}{\partial y_{m}} \sum \limits_{n=1}^{N} \mathrm{e}^{\mathrm{j} \frac{2\pi}{\lambda} \left(\|\mathbf{t}_{m,n}-\mathbf{s}_{i,l}\|_{2} - \|\mathbf{t}_{m,n}-\mathbf{s}_{j,\ell}\|_{2}\right)}\\
			&=\sum \limits_{n=1}^{N} \mathrm{j} \frac{2\pi}{\lambda} \left(\frac{y_{m}+[\mathbf{q}_{n}]_{2}-[\mathbf{s}_{i,l}]_{2}}{\|\mathbf{t}_{m}+\mathbf{q}_{n}-\mathbf{s}_{i,l}\|_{2}} 
			- \frac{y_{m}+[\mathbf{q}_{n}]_{2}-[\mathbf{s}_{j,\ell}]_{2}}{\|\mathbf{t}_{m}+\mathbf{q}_{n}-\mathbf{s}_{j,\ell}\|_{2}} \right)\\
			& ~~\times \mathrm{e}^{\mathrm{j} \frac{2\pi}{\lambda} \left(\|\mathbf{t}_{m,n}-\mathbf{s}_{i,l}\|_{2} - \|\mathbf{t}_{m,n}-\mathbf{s}_{j,\ell}\|_{2}\right)}, 1 \leq l \leq L_{i}, ~1 \leq \ell \leq L_{j}.
	\end{aligned}}
\end{equation*}

\section{Proof of Theorem \ref{Theo_ana}}\label{Appendix_ana}
For the case of a single channel path and $N=1$, we have $\mathbf{h}_{k}(\tilde{\mathbf{t}}) = b_{k}[\mathrm{e}^{\mathrm{j}\frac{2\pi}{\lambda} \|\mathbf{t}_{1}-\mathbf{s}_{k}\|_{2}},\mathrm{e}^{\mathrm{j}\frac{2\pi}{\lambda} \|\mathbf{t}_{2}-\mathbf{s}_{k}\|_{2}},\dots,\mathrm{e}^{\mathrm{j}\frac{2\pi}{\lambda} \|\mathbf{t}_{M}-\mathbf{s}_{k}\|_{2}}]^{\mathrm{T}}$. Thus, the equality in \eqref{eq_channel_power_one} always holds, i.e., $\left\|\mathbf{h}_{k}(\tilde{\mathbf{t}})\right\|_{1}^{2} = M^{2}N^{2}|b_{k}|^{2}$. To guarantee the phase alignment between the ABFV and all users' channel vectors, i.e., the equality in \eqref{eq_SNR_bound} to hold for $1 \leq k \leq K$, the sufficient as well as necessary condition can be derived as
\begin{equation}\label{eq_alignment_condi}{\small
	\begin{aligned}
		&\frac{\left|\mathbf{h}_{k}^{\mathrm{H}}(\tilde{\mathbf{t}}) \mathbf{w}\right|}{\left\|\mathbf{h}_{k}(\tilde{\mathbf{t}})\right\|_{2}} = 1, ~1 \leq k \leq K\\
		\Leftrightarrow&\frac{\left|\mathbf{h}_{k}^{\mathrm{H}}(\tilde{\mathbf{t}}) \mathbf{h}_{1}(\tilde{\mathbf{t}})\right|}{\left\|\mathbf{h}_{k}(\tilde{\mathbf{t}})\right\|_{2}\left\|\mathbf{h}_{1}(\tilde{\mathbf{t}})\right\|_{2}} = 1, ~2 \leq k \leq K\\
		\Leftrightarrow&\frac{1}{M}\sum \limits_{m=1}^{M}\mathrm{e}^{\mathrm{j}\frac{2\pi}{\lambda}\left(\left\|\mathbf{t}_{m}-\mathbf{s}_{1}\right\|_{2} -\left\|\mathbf{t}_{m}-\mathbf{s}_{k}\right\|_{2}\right)}=\mathrm{e}^{\mathrm{j}\phi_{k}}, ~2 \leq k \leq K\\
		\Leftrightarrow&\left\|\mathbf{t}_{m}-\mathbf{s}_{1}\right\|_{2} -\left\|\mathbf{t}_{m}-\mathbf{s}_{k}\right\|_{2} = \lambda(n_{k}+\phi_{k}), ~2 \leq k \leq K.
	\end{aligned}}
\end{equation}
This thus completes the proof.

\section{Derivation of Gradient $\nabla_{\tilde{\mathbf{t}}} \eta(\tilde{\mathbf{t}}, \mathbf{w})$}\label{Appendix_SNR_grad_APV}
The gradient $\nabla_{\tilde{\mathbf{t}}} \gamma(\tilde{\mathbf{t}})$ is composed of the partial derivative of $\gamma(\tilde{\mathbf{t}})$ w.r.t. $x_{m}$ and $y_{m}$, $1 \leq m \leq M$. Specifically, we have
\begin{equation}\label{eq_SNR_grad_x}{\small
	\begin{aligned}
		&\frac{\partial \eta(\tilde{\mathbf{t}}, \mathbf{w})}{\partial x_{m}} = \frac{P}{\sigma^{2}}\left( \sum \limits_{k=1}^{K} \frac{1}{\left|\mathbf{h}_{k}^{\mathrm{H}}(\tilde{\mathbf{t}}) \mathbf{w}\right|^{2}}\right)^{-2} \times \\
		&~~~~~~~~~~~~~~~~\sum \limits_{k=1}^{K} \frac{1}{\left|\mathbf{h}_{k}^{\mathrm{H}}(\tilde{\mathbf{t}}) \mathbf{w}\right|^{4}} \times \frac{\partial \mathbf{w}^{\mathrm{H}}\mathbf{h}_{k}(\tilde{\mathbf{t}})  \mathbf{h}_{k}^{\mathrm{H}}(\tilde{\mathbf{t}})\mathbf{w}}{\partial x_{m}},
	\end{aligned}}
\end{equation}
with 
\begin{equation*}{\small
		\begin{aligned}
			&\frac{\partial \mathbf{w}^{\mathrm{H}}\mathbf{h}_{k}(\tilde{\mathbf{t}})  \mathbf{h}_{k}^{\mathrm{H}}(\tilde{\mathbf{t}})\mathbf{w}}{\partial x_{m}}=\\
			&2\Re\Big{\{}\mathbf{w}^{\mathrm{H}}\mathbf{h}_{k}(\tilde{\mathbf{t}}) \sum \limits_{\ell = 1}^{L_{k}}\sum \limits_{n = 1}^{N} [\mathbf{b}_{k}]_{\ell} [\mathbf{w}]_{(m-1)N+n} \times \\
			&~~~~~~\mathrm{j}\frac{2\pi}{\lambda}\frac{x_{m}+[\mathbf{q}_{n}]_{1}-[\mathbf{s}_{k,\ell}]_{1}}{\|\mathbf{t}_{m}+\mathbf{q}_{n}-\mathbf{s}_{k,\ell}\|_{2}}
			\mathrm{e}^{\mathrm{j}\frac{2\pi}{\lambda}\|\mathbf{t}_{m,n}-\mathbf{s}_{k,\ell}\|}\Big{\}}.
	\end{aligned}}
\end{equation*}
Similarly, we have
\begin{equation}\label{eq_SNR_grad_y}{\small
		\begin{aligned}
			&\frac{\partial \eta(\tilde{\mathbf{t}}, \mathbf{w})}{\partial y_{m}} = \frac{P}{\sigma^{2}}\left( \sum \limits_{k=1}^{K} \frac{1}{\left|\mathbf{h}_{k}^{\mathrm{H}}(\tilde{\mathbf{t}}) \mathbf{w}\right|^{2}}\right)^{-2} \times \\
			&~~~~~~~~~~~~~~~~\sum \limits_{k=1}^{K} \frac{1}{\left|\mathbf{h}_{k}^{\mathrm{H}}(\tilde{\mathbf{t}}) \mathbf{w}\right|^{4}} \times \frac{\partial \mathbf{w}^{\mathrm{H}}\mathbf{h}_{k}(\tilde{\mathbf{t}})  \mathbf{h}_{k}^{\mathrm{H}}(\tilde{\mathbf{t}})\mathbf{w}}{\partial y_{m}},
	\end{aligned}}
\end{equation}
with 
\begin{equation*}{\small
		\begin{aligned}
			&\frac{\partial \mathbf{w}^{\mathrm{H}}\mathbf{h}_{k}(\tilde{\mathbf{t}})  \mathbf{h}_{k}^{\mathrm{H}}(\tilde{\mathbf{t}})\mathbf{w}}{\partial y_{m}}=\\
			&2\Re\Big{\{}\mathbf{w}^{\mathrm{H}}\mathbf{h}_{k}(\tilde{\mathbf{t}}) \sum \limits_{\ell = 1}^{L_{k}}\sum \limits_{n = 1}^{N} [\mathbf{b}_{k}]_{\ell}^{*} [\mathbf{w}]_{(m-1)N+n} \times \\
			&~~~~~~\mathrm{j}\frac{2\pi}{\lambda}\frac{y_{m}+[\mathbf{q}_{n}]_{2}-[\mathbf{s}_{k,\ell}]_{2}}{\|\mathbf{t}_{m}+\mathbf{q}_{n}-\mathbf{s}_{k,\ell}\|_{2}}
			\mathrm{e}^{\mathrm{j}\frac{2\pi}{\lambda}\|\mathbf{t}_{m,n}-\mathbf{s}_{k,\ell}\|}\Big{\}}.
	\end{aligned}}
\end{equation*}

\section{Derivation of Gradient $\nabla_{\boldsymbol{\varphi}} \eta(\tilde{\mathbf{t}}, \mathrm{e}^{\mathrm{j}\boldsymbol{\varphi}}/\sqrt{MN})$}\label{Appendix_SNR_grad_ABFV}
The $q$-th component of $\nabla_{\boldsymbol{\varphi}} \eta(\tilde{\mathbf{t}}, \mathrm{e}^{\mathrm{j}\boldsymbol{\varphi}}/\sqrt{MN})$, $1 \leq q \leq MN$, is given by
\begin{equation}\label{eq_SNR_grad_w}{\small
		\begin{aligned}
			&\frac{\partial \eta(\tilde{\mathbf{t}}, \mathrm{e}^{\mathrm{j}\boldsymbol{\varphi}}/\sqrt{MN})}
			{\partial [\boldsymbol{\varphi}]_{q}} 
			= \frac{P}{\sigma^{2}}\left( \sum \limits_{k=1}^{K} \frac{MN}{\left|\mathbf{h}_{k}^{\mathrm{H}}(\tilde{\mathbf{t}}) \mathrm{e}^{\mathrm{j}\boldsymbol{\varphi}}\right|^{2}}\right)^{-2} \times \\
			&~~\sum \limits_{k=1}^{K} \frac{M^{2}N^{2}}{\left|\mathbf{h}_{k}^{\mathrm{H}}(\tilde{\mathbf{t}}) \mathrm{e}^{\mathrm{j}\boldsymbol{\varphi}}\right|^{4}} \times \frac{\partial {(\mathrm{e}^{\mathrm{j}\boldsymbol{\varphi}})}^{\mathrm{H}}\mathbf{h}_{k}(\tilde{\mathbf{t}})  \mathbf{h}_{k}^{\mathrm{H}}(\tilde{\mathbf{t}})\mathrm{e}^{\mathrm{j}\boldsymbol{\varphi}}}{\partial [\boldsymbol{\varphi}]_{q}},
	\end{aligned}}
\end{equation}
with 
\begin{equation*}{\small
		\begin{aligned}
			&\frac{\partial {(\mathrm{e}^{\mathrm{j}\boldsymbol{\varphi}})}^{\mathrm{H}}\mathbf{h}_{k}(\tilde{\mathbf{t}})  \mathbf{h}_{k}^{\mathrm{H}}(\tilde{\mathbf{t}})\mathrm{e}^{\mathrm{j}\boldsymbol{\varphi}}}{\partial [\boldsymbol{\varphi}]_{q}}
			=2\Re\{-\mathrm{j}\mathrm{e}^{-\mathrm{j}[\boldsymbol{\varphi}]_{q}}[\mathbf{h}_{k}(\tilde{\mathbf{t}})]_{q} \mathbf{h}_{k}^{\mathrm{H}}(\tilde{\mathbf{t}})\mathrm{e}^{\mathrm{j}\boldsymbol{\varphi}}\}.
	\end{aligned}}
\end{equation*}

\bibliographystyle{IEEEtran} 
\bibliography{IEEEabrv,ref_zhu}

\begin{thebibliography}{10}
\providecommand{\url}[1]{#1}
\csname url@samestyle\endcsname
\providecommand{\newblock}{\relax}
\providecommand{\bibinfo}[2]{#2}
\providecommand{\BIBentrySTDinterwordspacing}{\spaceskip=0pt\relax}
\providecommand{\BIBentryALTinterwordstretchfactor}{4}
\providecommand{\BIBentryALTinterwordspacing}{\spaceskip=\fontdimen2\font plus
\BIBentryALTinterwordstretchfactor\fontdimen3\font minus
  \fontdimen4\font\relax}
\providecommand{\BIBforeignlanguage}[2]{{%
\expandafter\ifx\csname l@#1\endcsname\relax
\typeout{** WARNING: IEEEtran.bst: No hyphenation pattern has been}%
\typeout{** loaded for the language `#1'. Using the pattern for}%
\typeout{** the default language instead.}%
\else
\language=\csname l@#1\endcsname
\fi
#2}}
\providecommand{\BIBdecl}{\relax}
\BIBdecl

\bibitem{Paulraj2004Anover}
A.~Paulraj, D.~Gore, R.~Nabar, and H.~Bolcskei, ``An overview of {MIMO}
  communications - a key to gigabit wireless,'' \emph{Proc. IEEE}, vol.~92,
  no.~2, pp. 198--218, Feb. 2004.

\bibitem{Larsson2014massive}
E.~G. Larsson, O.~Edfors, F.~Tufvesson, and T.~L. Marzetta, ``Massive {MIMO}
  for next generation wireless systems,'' \emph{IEEE Commun. Mag.}, vol.~52,
  no.~2, pp. 186--195, Feb. 2014.

\bibitem{you2024next}
C.~You, Y.~Cai, Y.~Liu, M.~Di~Renzo, T.~M. Duman, A.~Yener, and A.~L.
  Swindlehurst, ``Next generation advanced transceiver technologies for {6G}
  and beyond,'' \emph{arXiv preprint arXiv:2403.16458}, 2024.

\bibitem{lu2024nearXL}
H.~Lu, Y.~Zeng, C.~You, Y.~Han, J.~Zhang, Z.~Wang, Z.~Dong, S.~Jin, C.-X. Wang,
  T.~Jiang, X.~You, and R.~Zhang, ``A tutorial on near-field {XL-MIMO}
  communications towards {6G},'' \emph{IEEE Commun. Surveys Tuts.}, 2024, early
  access, DOI: 10.1109/COMST.2024.3387749.

\bibitem{Wang2024XLMIMO}
Z.~Wang, J.~Zhang, H.~Du, D.~Niyato, S.~Cui, B.~Ai, M.~Debbah, K.~B. Letaief,
  and H.~V. Poor, ``A tutorial on extremely large-scale {MIMO} for {6G}:
  Fundamentals, signal processing, and applications,'' \emph{IEEE Commun.
  Surveys Tuts.}, vol.~26, no.~3, pp. 1560--1605, 3rd Quart. 2024.

\bibitem{zhu2023MAMag}
L.~Zhu, W.~Ma, and R.~Zhang, ``Movable antennas for wireless communication:
  Opportunities and challenges,'' \emph{IEEE Commun. Mag.}, vol.~62, no.~6, pp.
  114--120, June 2024.

\bibitem{wong2022bruce}
K.-K. Wong, K.-F. Tong, Y.~Shen, Y.~Chen, and Y.~Zhang, ``Bruce lee-inspired
  fluid antenna system: Six research topics and the potentials for {6G},''
  \emph{Front. Comms. Net.}, vol.~3, no. 853416, pp. 1--31, Mar. 2022.

\bibitem{zheng2024flexible}
J.~Zheng, J.~Zhang, H.~Du, D.~Niyato, S.~Sun, B.~Ai, and K.~B. Letaief,
  ``Flexible-position {MIMO} for wireless communications: Fundamentals,
  challenges, and future directions,'' \emph{IEEE Wireless Commun.}, 2024,
  early access, DOI: 10.1109/MWC.011.2300428.

\bibitem{zhu2022MAmodel}
L.~Zhu, W.~Ma, and R.~Zhang, ``Modeling and performance analysis for movable
  antenna enabled wireless communications,'' \emph{IEEE Trans. Wireless
  Commun.}, vol.~23, no.~6, pp. 6234--6250, June 2024.

\bibitem{ma2022MAmimo}
W.~Ma, L.~Zhu, and R.~Zhang, ``{MIMO} capacity characterization for movable
  antenna systems,'' \emph{IEEE Trans. Wireless Commun.}, vol.~23, no.~4, pp.
  3392--3407, Apr. 2024.

\bibitem{Wong2023opport}
K.-K. Wong, K.-F. Tong, Y.~Chen, Y.~Zhang, and C.-B. Chae, ``Opportunistic
  fluid antenna multiple access,'' \emph{IEEE Trans. Wireless Commun.},
  vol.~22, no.~11, pp. 7819--7833, Nov. 2023.

\bibitem{New2024fluid}
W.~K. New, K.-K. Wong, H.~Xu, K.-F. Tong, and C.-B. Chae, ``Fluid antenna
  system: New insights on outage probability and diversity gain,'' \emph{IEEE
  Trans. Wireless Commun.}, vol.~23, no.~1, pp. 128--140, Jan. 2024.

\bibitem{ning2024movable}
B.~Ning, S.~Yang, Y.~Wu, P.~Wang, W.~Mei, C.~Yuen, and E.~Bj{\"o}rnson,
  ``Movable antenna-enhanced wireless communications: General architectures and
  implementation methods,'' \emph{arXiv preprint arXiv:2407.15448}, 2024.

\bibitem{shao20246DMANet}
X.~Shao and R.~Zhang, ``{6DMA} enhanced wireless network with flexible antenna
  position and rotation: Opportunities and challenges,'' \emph{arXiv preprint
  arXiv:2406.06064}, 2024.

\bibitem{shao20246DMA}
X.~Shao, Q.~Jiang, and R.~Zhang, ``{6D} movable antenna based on user
  distribution: Modeling and optimization,'' \emph{arXiv preprint
  arXiv:2403.08123}, 2024.

\bibitem{paracha2019liquid}
K.~N. Paracha \emph{et~al.}, ``Liquid metal antennas: Materials, fabrication
  and applications,'' \emph{Sensors}, vol.~20, no.~1, pp. 1--26, Dec. 2019.

\bibitem{Mitha2022DPCA}
T.~H. Mitha and M.~Pour, ``Sidelobe reductions in linear array antennas using
  electronically displaced phase center antenna technique,'' \emph{IEEE Trans.
  Antennas Propagat.}, vol.~70, no.~6, pp. 4369--4378, Jun. 2022.

\bibitem{Alrabadi2013dense}
O.~N. Alrabadi, E.~Tsakalaki, H.~Huang, and G.~F. Pedersen, ``Beamforming via
  large and dense antenna arrays above a clutter,'' \emph{IEEE J. Select. Areas
  Commun.}, vol.~31, no.~2, pp. 314--325, Feb. 2013.

\bibitem{zhu2024historical}
L.~Zhu and K.-K. Wong, ``Historical review of fluid antenna and movable
  antenna,'' \emph{arXiv preprint arXiv:2401.02362}, 2024.

\bibitem{zhu2024wideband}
L.~Zhu, W.~Ma, Z.~Xiao, and R.~Zhang, ``Performance analysis and optimization
  for movable antenna aided wideband communications,'' \emph{arXiv preprint
  arXiv:2401.08974}, 2024.

\bibitem{mei2024movable}
W.~Mei, X.~Wei, B.~Ning, Z.~Chen, and R.~Zhang, ``Movable-antenna position
  optimization: A graph-based approach,'' \emph{IEEE Wireless Commun. Lett.},
  vol.~13, no.~7, pp. 1853--1857, Jul. 2024.

\bibitem{New2024MIMOFAS}
W.~K. New, K.-K. Wong, H.~Xu, K.-F. Tong, and C.-B. Chae, ``An
  information-theoretic characterization of {MIMO-FAS}: Optimization,
  diversity-multiplexing tradeoff and q-outage capacity,'' \emph{IEEE Trans.
  Wireless Commun.}, vol.~23, no.~6, pp. 5541--5556, Jun. 2024.

\bibitem{chen2023joint}
X.~Chen, B.~Feng, Y.~Wu, D.~W.~K. Ng, and R.~Schober, ``Joint beamforming and
  antenna movement design for moveable antenna systems based on statistical
  {CSI},'' in \emph{Proc. IEEE Global Commun. Conf.}, Kuala Lumpur, Malaysia,
  Dec. 2023, pp. 4387--4392.

\bibitem{zhu2023MAmultiuser}
L.~Zhu, W.~Ma, B.~Ning, and R.~Zhang, ``Movable-antenna enhanced multiuser
  communication via antenna position optimization,'' \emph{IEEE Trans. Wireless
  Commun.}, vol.~23, no.~7, pp. 7214--7229, July 2024.

\bibitem{xiao2023multiuser}
Z.~Xiao, X.~Pi, L.~Zhu, X.-G. Xia, and R.~Zhang, ``Multiuser communications
  with movable-antenna base station: Joint antenna positioning, receive
  combining, and power control,'' \emph{arXiv preprint arXiv:2308.09512}, 2023.

\bibitem{wu2024globallyMA}
Y.~Wu, D.~Xu, D.~W.~K. Ng, W.~Gerstacker, and R.~Schober, ``Globally optimal
  movable antenna-enhanced multi-user communication: Discrete antenna
  positioning, motion power consumption, and imperfect {CSI},'' \emph{arXiv
  preprint arXiv:2408.15435}, 2024.

\bibitem{hu2024power}
G.~Hu, Q.~Wu, K.~Xu, J.~Ouyang, J.~Si, Y.~Cai, and N.~Al-Dhahir, ``Fluid
  antennas-enabled multiuser uplink: A low-complexity gradient descent for
  total transmit power minimization,'' \emph{IEEE Commun. Lett.}, vol.~28,
  no.~3, pp. 602--606, Mar. 2024.

\bibitem{Yang2024movable}
S.~Yang, W.~Lyu, B.~Ning, Z.~Zhang, and C.~Yuen, ``Flexible precoding for
  multi-user movable antenna communications,'' \emph{IEEE Wireless Commun.
  Lett.}, vol.~13, no.~5, pp. 1404--1408, May 2024.

\bibitem{zhou2024MANOMA}
Y.~Zhou, W.~Chen, Q.~Wu, X.~Zhu, and N.~Cheng, ``Movable antenna empowered
  downlink {NOMA} systems: Power allocation and antenna position
  optimization,'' \emph{IEEE Wireless Commun. Lett.}, 2024, early access,
  DOI:10.1109/LWC.2024.3445110.

\bibitem{zhu2023MAarray}
L.~Zhu, W.~Ma, and R.~Zhang, ``Movable-antenna array enhanced beamforming:
  Achieving full array gain with null steering,'' \emph{IEEE Commun. Lett.},
  vol.~27, no.~12, pp. 3340--3344, Dec. 2023.

\bibitem{ma2024multi}
W.~Ma, L.~Zhu, and R.~Zhang, ``Multi-beam forming with movable-antenna array,''
  \emph{IEEE Commun. Lett.}, vol.~28, no.~3, pp. 697--701, Mar. 2024.

\bibitem{zhu2024dynamic}
L.~Zhu, X.~Pi, W.~Ma, Z.~Xiao, and R.~Zhang, ``Dynamic beam coverage for
  satellite communications aided by movable-antenna array,'' \emph{arXiv
  preprint arXiv:2404.15643}, 2024.

\bibitem{liu2024MAUAV}
W.~Liu, X.~Zhang, H.~Xing, J.~Ren, Y.~Shen, and S.~Cui, ``{UAV}-enabled
  wireless networks with movable-antenna array: Flexible beamforming and
  trajectory design,'' \emph{arXiv preprint arXiv:2405.20746}, 2024.

\bibitem{hu2024secure}
G.~Hu, Q.~Wu, K.~Xu, J.~Si, and N.~Al-Dhahir, ``Secure wireless communication
  via movable-antenna array,'' \emph{IEEE Signal Process. Lett.}, vol.~31, pp.
  516--520, Jan. 2024.

\bibitem{tang2024secure}
J.~Tang, C.~Pan, Y.~Zhang, H.~Ren, and K.~Wang, ``Secure {MIMO} communication
  relying on movable antennas,'' \emph{arXiv preprint arXiv:2403.04269}, 2024.

\bibitem{ding2024fullduplex}
J.~Ding, Z.~Zhou, W.~Li, C.~Wang, L.~Lin, and B.~Jiao, ``Movable
  antenna-enabled co-frequency co-time full-duplex wireless communication,''
  \emph{IEEE Commun. Lett.}, 2024, early access, DOI:
  10.1109/LCOMM.2024.3453296.

\bibitem{ma2024sensing}
W.~Ma, L.~Zhu, and R.~Zhang, ``Movable antenna enhanced wireless sensing via
  antenna position optimization,'' \emph{IEEE Trans. Wireless Commun.}, 2024,
  early access, DOI: 10.1109/TWC.2024.3443293.

\bibitem{lyu2024flexible}
W.~Lyu, S.~Yang, Y.~Xiu, Z.~Zhang, C.~Assi, and C.~Yuen, ``Flexible beamforming
  for movable antenna-enabled integrated sensing and communication,''
  \emph{arXiv preprint arXiv:2405.10507}, 2024.

\bibitem{liu2023near}
Y.~Liu, Z.~Wang, J.~Xu, C.~Ouyang, X.~Mu, and R.~Schober, ``Near-field
  communications: A tutorial review,'' \emph{IEEE Open J. Commun. Society},
  vol.~4, pp. 1999--2049, Sep. 2023.

\bibitem{chen2024joint}
Y.~Chen, M.~Chen, H.~Xu, Z.~Yang, K.-K. Wong, and Z.~Zhang, ``Joint beamforming
  and antenna design for near-field fluid antenna system,'' \emph{arXiv
  preprint arXiv:2407.05791}, 2024.

\bibitem{zhu2024suppressing}
Y.~Zhu, Q.~Wu, Y.~Liu, Q.~Shi, and W.~Chen, ``Suppressing beam squint effect
  for near-field wideband communication through movable antennas,'' \emph{arXiv
  preprint arXiv:2407.19511}, 2024.

\bibitem{ding2024near}
J.~Ding, L.~Zhu, Z.~Zhou, B.~Jiao, and R.~Zhang, ``Near-field multiuser
  communications aided by movable antennas,'' \emph{arXiv preprint
  arXiv:2408.10552}, 2024.

\bibitem{boyd2004convex}
S.~Boyd and L.~Vandenberghe, \emph{Convex {O}ptimization}.\hskip 1em plus 0.5em
  minus 0.4em\relax Cambridge, U.K.: Cambridge Univ. Press, 2004.

\end{thebibliography}

\end{document}